\documentclass{article}

\usepackage[preprint]{neurips_data_2024}
\makeatletter\renewcommand{\@notice}{}\makeatother

\usepackage[utf8]{inputenc}
\usepackage[T1]{fontenc}
\usepackage{url}
\usepackage{booktabs}
\usepackage{multirow}
\usepackage{amsfonts}
\usepackage{nicefrac}
\usepackage{microtype}
\usepackage{xcolor}
\usepackage{amsmath,amssymb}
\usepackage{amsthm}
\usepackage{graphicx}
\usepackage{float}
\usepackage{tikz}
\usetikzlibrary{arrows.meta,positioning,calc,backgrounds,fit,shapes.geometric}
\newtheorem{definition}{Definition}
\newtheorem{proposition}{Proposition}
\newtheorem{remark}{Remark}

\title{DSGE as a Structured World Model:\\
Benchmarking Counterfactual Generalization in Economic Worlds}

\author{%
  Wenli Xu \\
  Faculty of Finance \\
  City University of Macau \\
  \texttt{wlxu@cityu.edu.mo} \\
}

\usepackage{hyperref}

\begin{document}
\maketitle

\begin{abstract}
Modern world models---Dreamer, transformer world models (IRIS, Genie), and JEPA /
next-latent architectures---learn dynamics from observed trajectories but share a
weakness: their transition map is disciplined only where data were seen, so it degrades
under \emph{policy-induced distribution shift} and on \emph{counterfactual states} off
the training path. We argue that a Dynamic Stochastic General Equilibrium (DSGE) model is
a \emph{structured world model}: its state is a \emph{belief state}---the very object a
latent world model learns, but supplied with causal structure and hard cross-equation
constraints. We introduce \textbf{DSGE-Gym}, a benchmark of eight DSGE environments with
off-path counterfactual test sets, scaling to the ECB's 230-variable New Area-Wide Model.
We find that (i)~learned world models match the dynamics on-path but collapse off-path
(5$\sigma$ tail RMSE up to ${\sim}40\times$ the on-path level), and (ii)~training the
\emph{same} architectures on data the DSGE \emph{generates} across rare and
counterfactual-policy states---coverage only a structural model can synthesize---roughly
halves tail error and cuts policy-regime error $10$--$280\times$ where the counterfactual
rule shifts the ergodic support. Because such coverage cannot be sampled from any single
history, this measures structure's ability to \emph{manufacture} the missing
distribution. DSGE-Gym and all code are released as a reproducible testbed for
counterfactual generalization.
\end{abstract}

\begin{figure}[H]
\centering
\vspace{-2pt}
\resizebox{0.62\linewidth}{!}{%
\begin{tikzpicture}[
  font=\small,
  box/.style={draw,rounded corners=2pt,minimum height=8mm,minimum width=11mm,align=center},
  lat/.style={draw,circle,minimum size=9mm,inner sep=1pt,align=center},
  ->,>=Latex,thick,node distance=6mm]

\node[align=center] at (3.15,2.25) {\textbf{(a) The world-model loop}\\[-1pt]\footnotesize both learn the same object};
\node (o)    at (0,0)   {$o_t$};
\node[box]      (enc) [right=7mm of o]   {Enc};
\node[lat]      (z)   [right=7mm of enc]  {$z_t$};
\node[box]      (pred)[right=10mm of z]    {$\Gamma$};
\node[lat]      (zn)  [right=10mm of pred] {$\hat z_{t+1}$};
\draw (o) -- (enc);
\draw (enc) -- (z);
\draw (z) -- (pred);
\draw (pred) -- (zn);
\node[gray!55!black,font=\footnotesize,below=1.5mm of z] {belief state};
\node (a) [below=13mm of pred] {$a_t$\,\footnotesize(policy)};
\draw (a) -- (pred);
\draw[dashed,->] (zn) to[out=60,in=120] node[above,pos=0.5]{\footnotesize rollout} (z);
\node[align=center,text width=56mm] at (3.15,-2.9)
 {\footnotesize \textbf{Learned:} $\Gamma$ \emph{fit} from trajectories, disciplined only on the data support.\\
  \textbf{DSGE:} $\Gamma,\Phi$ \emph{deduced} from theory, defined at \emph{every} state.};

\begin{scope}[xshift=88mm]
\node[align=center] at (2.6,2.25) {\textbf{(b) On-path vs.\ off-path}\\[-1pt]\footnotesize where the two diverge};
\draw[thick,fill=green!5] (-0.3,-1.7) rectangle (5.6,1.7);
\node[green!45!black,anchor=north west] at (-0.2,1.6) {\footnotesize DSGE: valid everywhere};
\begin{scope}
\clip (-0.3,-1.7) rectangle (5.6,1.7);
\fill[blue!12] (1.5,-0.1) ellipse (1.4 and 0.85);
\draw[blue!55!black,dashed,thick] (1.5,-0.1) ellipse (1.4 and 0.85);
\end{scope}
\node[blue!45!black] at (1.5,-0.1) {\footnotesize on-path};
\node[blue!45!black,align=center,font=\scriptsize] at (1.5,-1.45) {learned WM\\trained here};
\foreach \p in {(1.05,0.12),(1.9,0.18),(1.4,-0.35),(2.1,-0.1),(0.95,-0.15),(1.7,0.3),(2.2,0.02),(1.25,0.28)}
  \fill[blue!55!black] \p circle (0.9pt);
\node[red!70!black] (tail) at (4.7,1.15) {$\bigstar$};
\node[red!70!black,font=\scriptsize,anchor=west] at (4.85,1.15) {tail};
\node[red!70!black] (reg) at (4.5,-1.25) {$\blacktriangle$};
\node[red!70!black,font=\scriptsize,anchor=west] at (4.65,-1.25) {regime};
\node[red!70!black,font=\scriptsize,align=center] at (3.9,0.75) {learned WM\\ breaks $\times$};
\end{scope}
\end{tikzpicture}%
}
\caption{\textbf{DSGE as a structured world model.} \textbf{(a)}~A learned world model and a
DSGE describe the same object---an encoder mapping observations to a \emph{belief state}
$z_t$ (the latent a world model is trained to learn) and a transition $\Gamma$ that rolls it
forward under a policy action $a_t$. They differ in origin: the learned map is \emph{fit}
from sampled trajectories, whereas the DSGE map $\Gamma$ and its cross-equation constraints
$\Phi$ are \emph{deduced} from theory. \textbf{(b)}~That difference bites off-path. A learned
model is disciplined only on the on-path ergodic region it was trained on (blue); on
\emph{tail} states and counterfactual \emph{regime} shifts (red) it must extrapolate and
breaks. The DSGE is defined and internally consistent at every state, so it can not only
predict off-path but \emph{generate} the missing coverage---the mechanism we exploit and
benchmark in this paper.}
\label{fig:concept}
\end{figure}
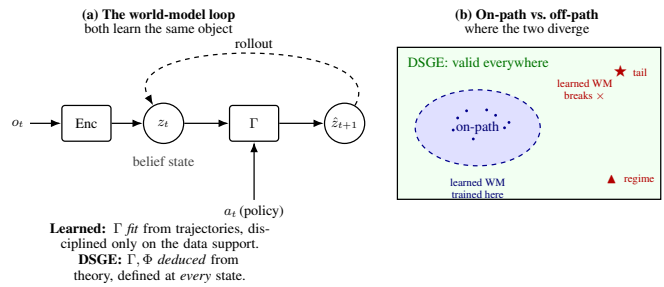

\clearpage
\section{Introduction}

Two research traditions have converged on the same object from opposite directions. In
machine learning, a \textbf{world model} is an internal predictive model of an
environment that an agent simulates to evaluate actions before taking them
\citep{ha2018world,hafner2020dream,lecun2022path}. In macroeconomics, a \textbf{DSGE
model} is a structural description of how an economy evolves under shocks and policy,
used for forecasting, counterfactual analysis, and policy evaluation
\citep{kydland1982time,smets2007shocks}. Both describe states, actions, transitions,
and futures; both exist so that an agent can \emph{imagine} consequences before acting.

\paragraph{Off-path generalization is the crux.}
A world model earns its keep precisely when an agent uses it to evaluate actions it has
\emph{not} yet taken, and thus to predict states the environment has \emph{not} yet
visited. A model accurate only on its training distribution is of little use for
planning, policy counterfactuals, or risk assessment---the very settings where the
decision-relevant states lie off the observed path. This is doubly true in economics,
where the questions that matter are inherently counterfactual: \emph{what if the central
bank changes its rule? what if a once-a-decade disaster hits? what if the same deficit is
financed by debt rather than taxes?} Evaluating a world model on held-out \emph{on-path}
data---the standard protocol---therefore measures the wrong thing. The right test is
off-path, and it is the test we build a benchmark around.

The machine-learning world model is \emph{learned}: a neural network fits a transition
$z_{t+1}=F(z_t,a_t)$ from observed trajectories. This recipe has driven remarkable
progress---latent imagination in the Dreamer family
\citep{hafner2019planet,hafner2020dream,hafner2021mastering,hafner2023dreamerv3},
sequence modeling of dynamics in transformer world models
\citep{micheli2023iris,bruce2024genie,alonso2024diamond}, and prediction in
representation space in the JEPA program
\citep{lecun2022path,assran2023ijepa,bardes2024vjepa,teoh2025nextlat}. Yet the fitted
map carries no guarantee away from the states it was trained on. Three failure modes
are decisive for economics:
\begin{enumerate}
\item \textbf{Policy-induced distribution shift.} A change in policy moves the
endogenous state distribution, so a transition fit under one policy is evaluated where
it was never trained---the machine-learning form of the Lucas critique
\citep{lucas1976critique}.
\item \textbf{Rare-event extrapolation.} Disasters, binding constraints, and tail risk
are exactly the states underrepresented on any observed path.
\item \textbf{Weak causal grounding.} A black-box transition need not respect
accounting identities, technology, or intertemporal optimality, so its off-path
predictions can be internally inconsistent.
\end{enumerate}

A DSGE model has the complementary profile. It is \emph{theory-given} rather than
data-fit; its parameters are structural and interpretable; and its equilibrium
conditions impose causal structure and exact constraints that hold \emph{everywhere in
the state space}, including states never observed. This contrast motivates our central
question:
\begin{quote}
\emph{Can economic theory itself serve as a world model---and does its structure buy
the counterfactual generalization that learned world models lack?}
\end{quote}

A recent and important step in this direction is Equilibrium World Models
\citep{schaab2026ewm}, which use deep learning to \emph{solve} DSGE models globally.
Our question is different and complementary: rather than building a better solver, we
treat the DSGE as the world model an agent plans with, and we ask how the structure it
carries affects counterfactual generalization---using the DSGE both as ground truth and
as a generator of training coverage.

\paragraph{What we find.}
On four environments spanning real, monetary, fiscal, and open-economy policy (RBC, NK,
TANK, TCM)---plus a 230-variable production model (NAWM) as a scale test---the results are
consistent and sharp. \emph{(i)~Learned world models collapse off-path.}
Models that match the structured dynamics on-path (normalized RMSE ${\approx}0.004$--
$0.08$) blow up under $5\sigma$ tail shocks, with the most flexible
sequence/attention models---the closest analogues of modern world models---reaching RMSE
${\approx}3$--$4$, two orders of magnitude worse; \emph{flexibility correlates with
off-path fragility}. \emph{(ii)~Structure recovers it through coverage.} Training the
\emph{same} architectures, at the \emph{same} sample size, on data the DSGE
\emph{generates} across rare states and counterfactual policy regimes recovers off-path
accuracy---robustly halving tail error for the expressive models, and $10$--$280\times$
lower error on policy-regime shifts that are unobservable in any real
dataset and that only a structural model can synthesize. The largest regime gains arise
wherever the counterfactual rule \emph{moves} the ergodic support---government spending in
RBC, and, replicating across very different mechanisms, the unemployment benefit in a
search-and-matching economy and the carbon tax in an environmental-NK economy
(\S\ref{sec:morenv})---whereas where the rule instead \emph{contracts} the support (NK,
TANK, TCM) the narrow-trained error is already small and the gain is correspondingly smaller
(\S\ref{sec:replication}).
\emph{(iii)~Approximate structure is not enough.} A first-order DSGE oracle and naive
static-equilibrium penalties do not close the gap, because they miss the intertemporal
dynamics that drive off-path behavior---which is exactly why the structure should be
\emph{learned}, not imposed wholesale. Together these isolate \emph{coverage generation}
as the concrete, measurable mechanism by which economic structure helps.

\paragraph{Contributions.}
\begin{enumerate}
\item \textbf{A reframing (\S\ref{sec:swm}).} We argue that DSGE is a \emph{structured
world model} and show that its state variable is a \emph{belief state} in the precise
sense used by recent transformer world models
\citep{teoh2025nextlat,hu2025belief}---the conceptual bridge that makes the two
literatures one.
\item \textbf{A benchmark (\S\ref{sec:bench}).} \textbf{DSGE-Gym}: ten DSGE
environments with a unified world-model interface and off-path counterfactual test
sets (tail shocks, policy-regime shifts, permanent shocks), built on a higher-order
global solver.
\item \textbf{A controlled study (\S\ref{sec:exp}).} Across four environments
(RBC, New Keynesian, the two-agent TANK economy, and a two-country open economy)---and, as
a scale test, the 230-variable ECB New Area-Wide Model---we confirm that learned world
models collapse off-path, and show that \emph{DSGE-generated coverage} of rare and
counterfactual states recovers off-path generalization for the same architectures---each
model exercising a different policy ``action'' (real, monetary, fiscal, open-economy). We
also
report an instructive negative result---a first-order oracle and naive structural
penalties do not suffice---which motivates \emph{learning} structured world models.
\end{enumerate}

\section{Related Work}
\label{sec:related}

\paragraph{Latent and recurrent world models.}
\citet{ha2018world} showed that a generative latent dynamics model can train a
controller entirely inside a learned ``dream.'' The Dreamer family built this into a
general online-RL paradigm: PlaNet \citep{hafner2019planet} introduced the recurrent
state-space model (RSSM); Dreamer \citep{hafner2020dream} learned behaviors by
backpropagating through imagined latent rollouts; DreamerV2 \citep{hafner2021mastering}
adopted discrete latents; and DreamerV3 \citep{hafner2023dreamerv3} achieved robust
performance across more than 150 tasks with fixed hyperparameters. These models compress
history into a compact recurrent state---by architecture, a sufficient statistic.

\paragraph{Transformer and generative world models.}
Casting dynamics as sequence modeling, IRIS \citep{micheli2023iris} learns an
autoregressive transformer over discrete image tokens and is highly sample-efficient on
Atari; Genie \citep{bruce2024genie} learns controllable, generative interactive
environments at scale; and DIAMOND \citep{alonso2024diamond} uses diffusion for visually
faithful world simulation. Transformers, however, retain a memory that grows with
sequence length and have no inherent pressure to compress history into a compact state
\citep{teoh2025nextlat}---the property our structural view supplies for free.

\paragraph{Joint-embedding and self-predictive architectures.}
LeCun's JEPA program \citep{lecun2022path} argues for predicting in representation space
rather than observation space; I-JEPA \citep{assran2023ijepa} and V-JEPA
\citep{bardes2024vjepa} instantiate this for images and video. Most relevant to us,
NextLat \citep{teoh2025nextlat} augments next-token training with a \emph{next-latent}
self-prediction loss and proves the hidden state converges to a \emph{belief state};
Belief State Transformers \citep{hu2025belief} target the same object. This line builds
on self-predictive representation learning in RL
\citep{tang2023understanding,ni2024bridging}. These works make ``the model's latent is a
sufficient statistic'' a concrete, trainable property---exactly the property a
structural state variable has by construction.

\paragraph{Belief states and sufficient statistics.}
A belief state is a sufficient statistic of history for predicting the future
\citep{kaelbling1998planning,striebel1965sufficient}. In economics the \emph{state
variable} of a recursive (Bellman) formulation is precisely such a statistic---the same
idea, formalized decades earlier and embedded in the rational-expectations equilibrium
concept.

\paragraph{Deep learning for dynamic economic models.}
A fast-growing literature solves high-dimensional dynamic models with neural networks:
Deep Equilibrium Nets \citep{azinovic2022deqn}, deep learning for dynamic economic
models \citep{maliar2021deep}, economics-inspired networks with stabilizing homotopies
\citep{azinovic2023homotopy}, and global solutions with endogenous wealth distributions
\citep{fernandez2023financial}. Equilibrium World Models \citep{schaab2026ewm} are the
closest prior work: they enforce exact equilibrium conditions on a broad, model-generated distribution of ordinary, rare, and counterfactual states and certify the
resulting policy. \textbf{Our question differs.} EWM and the above \emph{solve} a DSGE
(output: a policy function $\pi(x)$); we ask what happens when the DSGE \emph{is} the
world model an agent plans with, and study counterfactual generalization. EWM's coverage
insight reappears in \S\ref{sec:exp} as the operative mechanism by which structure
helps. Broader surveys of reinforcement learning for economics
\citep{rl4econ2026survey} situate this agenda within ML-for-economics.

\section{DSGE as a Structured World Model}
\label{sec:swm}

\subsection{Problem formulation}
\label{sec:formulation}
We work with a controlled stochastic process generating observations
$o_t\in\mathcal{O}$ under actions $a_t\in\mathcal{A}$, with history
$h_t=(o_{1:t},a_{1:t-1})$. Following the world-model literature we are interested in
\emph{compact} predictors of the future.

\begin{definition}[World model]
A \emph{world model} is a map $F$ and a state space $\mathcal{Z}$ such that the next
state is $z_{t+1}=F(z_t,a_t,\varepsilon_{t+1})$ with exogenous innovation
$\varepsilon_{t+1}$, and the future observation law factorizes through $z$, i.e.\
$p(o_{t+1}\mid h_t,a_t)=p(o_{t+1}\mid z_t,a_t)$.
\end{definition}

\begin{definition}[Belief state / sufficient statistic]
\label{def:belief}
A statistic $z_t=\sigma(h_t)$ is a \emph{belief state} if it is sufficient for the future
given the action sequence: for all $k\ge1$,
\[
p\!\left(o_{t+1:t+k}\mid h_t,a_{t:t+k-1}\right)=p\!\left(o_{t+1:t+k}\mid z_t,a_{t:t+k-1}\right).
\]
Equivalently $I(o_{>t};h_t\mid z_t,a_{\ge t})=0$: conditioning on the full history adds
no predictive information beyond $z_t$.
\end{definition}

This is the object recurrent and self-predictive world models are trained to recover
\citep{hafner2023dreamerv3,teoh2025nextlat,hu2025belief}: a finite code that screens off
history. A DSGE is exactly such an object in recursive (Bellman) form. With structural
state $x_t$ (predetermined and exogenous variables), control/policy $a_t$, and structural
shock $\varepsilon_t$, the rational-expectations equilibrium induces a law of motion
\begin{equation}
x_{t+1}=\Gamma(x_t,a_t,\varepsilon_{t+1}),\qquad
\mathbb{E}_t\,\Phi(x_{t+1},x_t,a_t)=0\ \ \text{for all }x_t,
\label{eq:dsge}
\end{equation}
where $\Phi$ collects the cross-equation restrictions---Euler equations, market clearing,
technology, policy rules---that hold at \emph{every} $x_t$, not only on the ergodic path.
$\Gamma$ is the policy function obtained by solving \eqref{eq:dsge}; in DSGE-Gym it is a
pruned third-order perturbation, exact up to the perturbation order.

\paragraph{Task and metric.}
We study the \emph{one-step world-model task}: learn $\hat F$ predicting $z_{t+1}$ from
$(z_t,\varepsilon_{t+1})$, scored by normalized RMSE on a split $\mathcal{D}$,
$\mathrm{RMSE}(\hat F;\mathcal{D})=\big(\mathbb{E}_{\mathcal{D}}\|\hat
F(z_t,\varepsilon_{t+1})-z_{t+1}\|_2^2\big)^{1/2}$ on $z$-scored variables. Writing
$\mathcal{D}_{\mathrm{on}}$ for the on-path ergodic distribution and
$\mathcal{D}_{\mathrm{off}}$ for an off-path (tail or counterfactual-regime) distribution,
the headline quantity is the \emph{off-path generalization gap}
\begin{equation}
\Delta(\hat F)=\mathrm{RMSE}(\hat F;\mathcal{D}_{\mathrm{off}})-\mathrm{RMSE}(\hat F;\mathcal{D}_{\mathrm{on}}).
\label{eq:gap}
\end{equation}
A world model is useful for planning and policy analysis precisely when $\Delta$ is small,
because the decision-relevant states are typically those absent from historical data.

\subsection{Two kinds of world model}
\label{sec:twokinds}
A learned world model and a DSGE describe the same tuple (states, actions, transitions,
futures) but obtain it by opposite routes, and the difference is exactly along the axes
that govern off-path behavior (Table~\ref{tab:twokinds}).

\paragraph{Origin and training signal.}
A learned world model \emph{induces} its transition from sampled trajectories: it sees
$(z_t,a_t,z_{t+1})$ tuples and minimizes a prediction loss. Its knowledge is therefore
bounded by the support of the data---it can interpolate where trajectories are dense and
must extrapolate everywhere else. A DSGE \emph{deduces} its transition from primitives
(preferences, technology, market clearing) and an equilibrium concept; the resulting map
$\Gamma$ and constraints $\Phi$ in \eqref{eq:dsge} are defined at \emph{every} point of
the state space, not only where data happen to fall.

\paragraph{Causality and constraints.}
Because the learned map fits a conditional expectation, it has no built-in reason to
respect accounting identities, technological feasibility, or intertemporal optimality;
these can be violated freely in regions with little training mass. The DSGE encodes them
as hard cross-equation restrictions that hold identically, so its predictions are
internally consistent even on never-seen states---the property that underwrites
counterfactual reasoning.

\paragraph{What each one buys.}
The learned model wins on \emph{flexibility, scale, and adaptivity}: it can absorb
high-dimensional observations, improve with more data, and---crucially---degrade
gracefully under misspecification, since it is not committed to a fixed structural form.
The DSGE wins on \emph{interpretability, causality, off-path discipline, and
generativity}: its parameters have economic meaning, its structure survives policy
change (\S\ref{sec:whyhelp}), and it can synthesize valid data anywhere. The two profiles
are mirror images---each is strong exactly where the other is weak---which is why we read
them not as competitors but as the two halves of a single design space
(\S\ref{sec:bridge}).

\begin{table}[t]
\centering
\caption{Two kinds of world model. The axes that differ are precisely those that govern
off-path generalization.}
\label{tab:twokinds}
\small
\begin{tabular}{lll}
\toprule
Dimension & DSGE & Learned (JEPA/Dreamer-style) \\
\midrule
Origin             & Deduced from theory + equilibrium   & Induced from sampled data \\
Training signal    & Cross-equation restrictions $\Phi$  & Prediction loss on trajectories \\
Interpretability   & Structural, economically meaningful & Black-box latent \\
Causality          & Strong causal structure             & Weak / correlational \\
Constraints        & Exact, hold at every state          & Soft, data-region only \\
Off-path behavior  & Disciplined by $\Phi$               & Unconstrained extrapolation \\
Coverage           & Generator: synthesizes any state    & Bounded by data support \\
Adaptivity         & Brittle to misspecification         & Improves with data \\
\bottomrule
\end{tabular}
\end{table}

\subsection{The bridge: a DSGE state variable is a belief state}
\label{sec:bridge}
The conceptual core of the paper is that the two literatures study the \emph{same}
object. We make this precise.

\begin{proposition}[The structural state is a belief state]
\label{prop:bridge}
Let $\{o_t\}$ be generated by a rational-expectations equilibrium with law of motion
\eqref{eq:dsge} and an observation map $o_t=g(x_t)$. Then the structural state $x_t$ is a
belief state in the sense of Definition~\ref{def:belief}. If $g$ is injective on the
ergodic set, $x_t$ is moreover a \emph{minimal} sufficient statistic.
\end{proposition}

\begin{proof}[Proof sketch]
By \eqref{eq:dsge} the process is first-order Markov in $x_t$: the conditional law of
$x_{t+1}$ given the entire history depends only on $(x_t,a_t)$ through $\Gamma$. Iterating,
the $k$-step-ahead law $p(x_{t+1:t+k}\mid h_t,a_{t:t+k-1})$ factors through $x_t$, and
since $o_{t+j}=g(x_{t+j})$ the same holds for observations, giving sufficiency
(Definition~\ref{def:belief}). Minimality: if $g$ is injective on the ergodic set then
$x_t$ is recoverable from $o_t$, so no coarser statistic of the history retains
sufficiency; hence $x_t$ is minimal. (When $g$ is not injective---partial
observability---the minimal sufficient statistic is the posterior over $x_t$, recovering
the classical POMDP belief state \citep{kaelbling1998planning}.)
\end{proof}

\begin{remark}
Proposition~\ref{prop:bridge} is, by design, elementary: it restates that the state of a
recursive (Markov) equilibrium is a sufficient statistic. Its role is not technical depth
but identification---it pins down, exactly, the object that recurrent and self-predictive
world models spend capacity \emph{learning} and that a DSGE supplies \emph{by
construction}. The contribution of \S\ref{sec:bridge} is this correspondence, not the
lemma.
\end{remark}

Proposition~\ref{prop:bridge} says the DSGE state variable satisfies, \emph{by
construction}, exactly the property that NextLat \citep{teoh2025nextlat} trains a
transformer to acquire and that Belief State Transformers \citep{hu2025belief} target:
\begin{quote}
\emph{NextLat's belief state $\equiv$ a DSGE state variable.}
\end{quote}
The two literatures are therefore \emph{complementary}, not opposed. A learned world model
\emph{searches} for a compact sufficient statistic from data and must discover the
generating map; a DSGE \emph{hands you one}, together with the causal map $\Gamma$ that
generates it and the constraints $\Phi$ it satisfies. What the LeCun/JEPA world model
lacks---structure, causality, constraints---is precisely what a DSGE supplies; what the
DSGE lacks---adaptivity to data and graceful degradation under misspecification---is what
learning supplies. This complementarity is the seed of the second paper in our agenda:
learn a NextLat-style latent world model whose belief state is \emph{regularized} toward
$\Phi$, inheriting the data-fit of the transformer and the off-path discipline of the
structure.

\subsection{Why structure should help off-path}
\label{sec:whyhelp}
Three properties of \eqref{eq:dsge} bear directly on the off-path gap \eqref{eq:gap}.

\paragraph{(P1) Policy invariance (the Lucas critique, restated).}
Partition the parameters into deep parameters $\theta$ (preferences, technology) and a
policy rule indexed by $\psi$ (e.g.\ a Taylor coefficient, a fiscal rule). The decisive
structural assumption is that $\theta\perp\psi$: deep parameters do not move when the
rule changes.

\begin{proposition}[Re-composition under policy change]
\label{prop:lucas}
Suppose the equilibrium map factorizes as
$\Gamma_\psi(x,a,\varepsilon)=G\big(x,\pi_\psi(x),\varepsilon;\theta\big)$, where $G$
depends on $\psi$ only through the action $a=\pi_\psi(x)$. Then a world model that has
learned $G(\cdot;\theta)$ from regime $\psi$ predicts correctly under any new regime
$\psi'$ by substituting $\pi_{\psi'}$, whereas a black-box that has fit the
\emph{reduced-form} composition $\Gamma_\psi$ directly is biased off its training regime
by $\|G(x,\pi_{\psi'}(x),\cdot)-\Gamma_\psi(x,\cdot)\|$.
\end{proposition}

This is the Lucas critique \citep{lucas1976critique} as a statement about world-model
generalization: structural models re-compose under counterfactual policy; reduced-form
learners do not. It predicts exactly the failure we observe on the \textbf{regime} split.

\begin{remark}[Scope of Proposition~\ref{prop:lucas}]
\label{rem:lucas}
The clean factorization holds exactly for the transition of the \emph{predetermined}
states. For forward-looking \emph{jump} variables the one-step map under $\psi'$ depends on
agents' expectations of the entire future rule, so $G$ generally moves with $\psi$ beyond
the contemporaneous action; re-composition then requires the new equilibrium decision rule
$\pi_{\psi'}$ itself. That object is produced by re-solving the structural model and is
\emph{not} recoverable by a reduced-form learner trained on baseline-regime data. The
proposition should thus be read as bounding what \emph{any} single-regime learner can
achieve off-regime, and as identifying re-solving structure---not more data---as the only
admissible route to the counterfactual map. This is the precise sense in which our
\textbf{regime} results are about structure rather than extrapolation.
\end{remark}

\paragraph{(P2) Counterfactual consistency.}
The constraints $\Phi$ in \eqref{eq:dsge} hold at \emph{every} $x$, including states with
zero training mass. They therefore pin the map off-path where a likelihood-trained model
is unconstrained---a hard inductive bias active exactly where data are absent.

\paragraph{(P3) Coverage by construction.}
Because the DSGE is a \emph{generator}, it can synthesize valid transitions anywhere in
the state space---rare ($5\sigma$) states and entire counterfactual regimes that are
unobservable in real data. This turns structure into a \emph{data} advantage: a structural
model can teach a learned world model where history never went. \S\ref{sec:exp} isolates
(P3) as the operative mechanism in practice---and shows (P2), in the form of a fixed
linearized oracle or static penalties, is \emph{not} sufficient on its own, motivating
learned structure.

\section{DSGE-Gym: A Benchmark}
\label{sec:bench}

\paragraph{Environments.}
DSGE-Gym wraps a coherent suite of DSGE models with unified notation and a
shared world-model interface, plus a 230-variable production model as a scale test. This
paper fully specifies and benchmarks eight environments (Table~\ref{tab:envs}; equilibrium
systems, calibrations, and steady states in Appendix~\ref{app:models})---the four main
economies (RBC, NK, TANK, TCM), three further mechanisms not present in any of them
(\textbf{DMP} search-and-matching \emph{unemployment}, an \textbf{Environmental NK} economy
with a carbon stock and a productivity-damage feedback, and an endogenous \textbf{firm-entry}
economy with a product-variety margin, \S\ref{sec:morenv}), and the 230-variable NAWM scale
test. The released code additionally targets a further set of mechanisms---Gertler--Karadi
banks/QE, a non-tradable-goods open economy, and trend growth with permanent (unit-root)
shocks---so that the suite spans monetary, fiscal, labor, financial, climate, open-economy,
firm-dynamics, and non-stationary mechanisms, including heterogeneous-agent (TANK) economies
in the spirit of \citet{kaplan2018monetary}; these last three are distributed with the code
but are \emph{not} specified or benchmarked here, and enter no headline claim of this paper. Each environment is solved and simulated with a
pruned third-order perturbation solver---for genuine nonlinearity---and exposes
$z_{t+1}=F(z_t,a_t,\varepsilon_{t+1})$, where $z_t$ are model variables, $a_t$ the
policy/exogenous ``action,'' and $\varepsilon_t$ the structural shocks. The complete
equilibrium conditions, calibrations, steady states, and the explicit
$(z_t,a_t,\varepsilon_t)$ world-model mapping for every environment reported in the body
are given in Appendix~\ref{app:models}; the counterfactual test-set construction and all
implementation details are in Appendices~\ref{app:tests}--\ref{app:impl}.

\begin{table}[t]
\centering
\caption{Main DSGE-Gym environments fully specified and benchmarked in this paper.
Additional environments are distributed with the released code but are not specified or
benchmarked here.}
\label{tab:envs}
\begin{tabular}{ll}
\toprule
Model & Mechanism / counterfactual axis \\
\midrule
RBC          & Baseline real economy; government-spending action \\
New Keynesian & Monetary policy / Taylor-rule action; nominal rigidity \\
TANK         & Household heterogeneity, hand-to-mouth agents, fiscal-financing action \\
Two-country (TCM) & Open economy, real exchange rate, international bonds; monetary action \\
DMP          & Search-and-matching unemployment; labor-market (UI-benefit) action \\
E-NK         & Environmental NK; carbon stock + productivity damage; climate (carbon-tax) action \\
Firm entry   & Endogenous product variety (BGM); entry-cost / deregulation action \\
NAWM (230 vars) & Production-scale Euro Area--US model (ECB); scale stress test \\
\bottomrule
\end{tabular}
\end{table}

\paragraph{Unified interface.}
Every environment exposes the same API, so a baseline written once runs on all of them.
A model is registered by its variable and shock names; the simulator returns aligned
arrays of levels, and a thin loader forms supervised pairs
$x_t=[z_t,\varepsilon_{t+1}]$, $y_t=z_{t+1}$ (the realized next-period shock is supplied
as the ``action'' so the transition is near-deterministic and RMSE isolates the learned
map). Two reference oracles---a first-order analytic policy and the pruned third-order
policy---are exposed through the same interface, giving every task a structural upper
bound. The split generator (next paragraph) is identical across environments;
only the model file and its counterfactual parameters change. This is what lets
DSGE-Gym scale from an 11-variable RBC to the 230-variable NAWM---and to the full
ten-model suite---with no change to the learning code.

\paragraph{Tasks.}
DSGE-Gym defines three tasks of increasing difficulty.
\textbf{(T1) One-step prediction:} given $(z_t,\varepsilon_{t+1})$, predict $z_{t+1}$;
this is the task we study in depth and the cleanest probe of whether a model has learned
the transition. \textbf{(T2) Multi-step latent rollout:} autoregress the model $k$ steps
under a shock path and score accumulated error, which exposes compounding off-path drift
that one-step error can hide. \textbf{(T3) Lookahead planning:} choose an action sequence
$a_{t:t+H}$ to reach a target state, scored by \emph{regret} against a model-consistent
reference---the action sequence that the true dynamics $\Gamma$ imply reaches the target at
least cost under the model's own laws of motion (a control benchmark, not a welfare-optimal
Ramsey policy; the precise objective is fixed in the released T3 specification). T1 tests
the world model as a predictor; T2 tests it as a simulator; T3 tests it as a decision
tool---the use that makes off-path accuracy matter. Because the realized shock
$\varepsilon_{t+1}$ is supplied as input, \textbf{T1 is a deterministic-map regression}
that isolates whether a model has captured $\Gamma$; the distributional and multi-step
content that makes a world model a planning tool lives in T2/T3. We report T1 here and
release T2/T3 with the benchmark; the results below should be read as one-step map fidelity,
the necessary first probe rather than the whole of world-model quality.

\paragraph{Counterfactual test sets.}
The defining feature of DSGE-Gym is that train and test are drawn from the \emph{same}
solved model but different regions of its state space, so the gap between them is a clean
measurement of off-path generalization rather than a confound of data source. From each
model we generate (Appendix~\ref{app:tests}): \textbf{on-path} train/test (baseline
policy, ergodic $1\sigma$ shocks); \textbf{tail} ($5\sigma$ shocks, pushing into the
nonlinear region the third-order solution captures); \textbf{regime} (a counterfactual
change to a \emph{policy-rule} parameter---spending persistence in RBC, the Taylor
coefficient in NK, the fiscal-financing rule in TANK---evaluated by a model trained only
on the baseline regime, the operational form of the Lucas critique); and, where relevant,
\textbf{permanent} (non-stationary, unit-root) shocks. A fourth \textbf{wide} split is a
\emph{training} set the DSGE synthesizes across rare and counterfactual states at fixed
sample size, used to test the coverage hypothesis (P3).

\paragraph{Metrics.}
We report (i)~normalized RMSE per split (errors $z$-scored on on-path train statistics,
so units are comparable across variables and environments); (ii)~the \emph{off-path
generalization gap} $\Delta=\mathrm{RMSE}_{\mathrm{off}}-\mathrm{RMSE}_{\mathrm{on}}$
\eqref{eq:gap}, the headline quantity; and (iii)~a belief-state \emph{sufficiency probe}
that regresses the true structural state on a model's learned latent and reports
$R^2$---a direct test of whether the latent has become the sufficient statistic of
Definition~\ref{def:belief}. Reference oracles bound each metric from below. This paper
reports (i)--(ii); the sufficiency probe (iii) is most informative for \emph{learned}
latent world models with a non-trivial internal state, so we release it with the benchmark
and analyze it alongside the structured latent model of Paper~2 rather than report it on
the deterministic-map baselines here.

\section{Experiments}
\label{sec:exp}

\paragraph{Setup.}
We study four environments that share the same protocol but exercise different policy
``actions'': \textbf{RBC} (11 variables, 2 shocks; action = government spending),
\textbf{NK} (15 variables, 3 shocks; action = the monetary-policy rate via a Taylor
rule), \textbf{TANK} (22 variables, 3 shocks; two household types and public debt;
action = the fiscal-financing rule, under which Ricardian equivalence fails), and
\textbf{TCM} (45 variables, 6 shocks; a two-country open economy with a real exchange
rate, terms of trade, and international bond positions; action = the two Taylor rules).
Each is solved with a pruned third-order perturbation (quarterly calibrations from
\citealp{nispilandi2018rbc,nispilandi2021nk}; steady states verified against the
published values---e.g.\ TCM's non-analytic steady state, solved numerically, reproduces
all reported digits). The one-step
task predicts the next-period vector of model variables from the current vector plus the
realized next-period shocks; inputs/outputs are $z$-scored on on-path training
statistics. The \textbf{regime} (counterfactual) split changes a policy-rule parameter
that the learner never sees on-path---higher/persistent spending (RBC), a hawkish Taylor
coefficient $\phi_\pi{:}1.5{\to}3.0$ (NK), aggressive debt/spending stabilization
$\phi_d,\phi_g$ (TANK), or an exchange-rate-leaning rule $\phi_e{:}0{\to}0.5$
(TCM)---while the \textbf{tail} split applies $5\sigma$ shocks.
Baselines span the capacity spectrum: a \textbf{linear} (ridge) extrapolator, an
\textbf{MLP}, an \textbf{LSTM}, a \textbf{Transformer}, and \textbf{NextLat} (a
Transformer with the next-latent self-prediction loss of \citealp{teoh2025nextlat}). We
report means over three seeds. A fifth environment, the ECB New Area-Wide Model (NAWM;
230 variables, 21 shocks; \S\ref{sec:nawm}), serves as a production-scale stress test.
Implementation: MacroModelling.jl \citep{macromodelling}
for solving/simulation and PyTorch for the learned baselines. We present the RBC study in
detail (\S\ref{sec:exp}.1--3), then show the pattern replicates on NK and TANK
(\S\ref{sec:exp}.4).

\subsection{Learned world models collapse off-path (H1)}
Trained on on-path data, all learned models fit the dynamics well on-path (normalized
RMSE ${\approx}0.004$--$0.08$) but their error explodes off-path. Under 5$\sigma$ tail
shocks the attention/sequence models---the closest analogues of LeCun/JEPA-style world
models---reach RMSE ${\approx}3.3$--$3.8$, a gap of two orders of magnitude
(Table~\ref{tab:coverage}, ``narrow''); the low-capacity linear extrapolator degrades
far less. \emph{Flexibility correlates with off-path fragility}: precisely the models
with the most expressive learned dynamics fail hardest where data are absent. We do not
claim this is surprising in isolation---unconstrained over-parameterized predictors are
known to extrapolate poorly outside their training support, and the NAWM scale test
(\S\ref{sec:nawm}) makes this explicit by exhibiting the same collapse on a \emph{linear}
DGP that a linear baseline fits exactly. The contribution is not that learners extrapolate
badly but that a structural generator can supply the missing support (H2), and that
approximate imposed structure cannot (\S\ref{sec:exp}.2). Baselines are comparably and
modestly tuned across environments for protocol uniformity, so H1 should be read as the
behavior of standard learners under a common budget, not as a proof of an intrinsic ceiling
(see Limitations).

\subsection{A first-order DSGE oracle is a poor world model (negative result)}
A first-order (linearized) DSGE policy---a natural ``structured'' baseline---is a
\emph{local} Taylor expansion around the steady state. On the nonlinear (third-order)
data it already has nonzero on-path error and is \emph{worse than the black-box MLP}
off-path (tail RMSE ${\approx}0.80$ vs.\ $0.37$). Adding the static equilibrium
conditions as a soft penalty to the MLP also does not help, because those identities are
already satisfied on-path and do not pin the \emph{intertemporal} (Euler) dynamics that
drive off-path behavior. \textbf{Approximate structure is not enough}---which is exactly
why the next step is to \emph{learn} structured dynamics rather than impose a fixed
approximation.

\subsection{DSGE-generated coverage recovers off-path generalization (H2)}
The structural model's decisive advantage is that it is a \emph{generator}: it can
synthesize valid training data where real data never visit. We train the \emph{same}
architectures on a \textbf{wide} training set the DSGE produces across rare states
(1--5$\sigma$ shocks) and counterfactual policy regimes (varied spending and policy
persistence), holding the sample size fixed at the on-path level---so only
\emph{coverage} differs, not data quantity. Results are in Table~\ref{tab:coverage} and
Figure~\ref{fig:coverage}.

\paragraph{What ``coverage'' measures (and what it does not).}
We are deliberately explicit about the design. The wide set spans the same off-path
\emph{regions} the off-path tests probe---rare-shock scales and counterfactual policy
regimes (Appendix~\ref{app:tests})---and the evaluated \texttt{test\_regime} parameters lie
\emph{within} that span. H2 therefore does \emph{not} measure a learner extrapolating to a
regime it never saw; it measures what becomes possible once a structural generator
\emph{manufactures} the regime distribution that no single history contains. The scientific
content is the \emph{generativity} of structure (Remark~\ref{rem:lucas}: re-solving is the
only admissible route to the counterfactual map), not generalization of the learner. We
read the regime numbers in that light, and flag two controls we do not yet run but release
with the benchmark: (a)~a stricter \emph{leave-one-regime-out} sweep (train on a band of
regimes, test on a withheld one), and (b)~a \emph{non-structural} coverage source
(regime-dummy VARs, parameter-perturbation augmentation) that would isolate equilibrium
structure from coverage \emph{per se}. Both sharpen the claim and are deferred to Paper~2
(see Limitations).

\begin{table}[t]
\centering
\caption{Off-path RMSE, narrow (on-path) vs.\ wide (DSGE-coverage) training
(mean over 3 seeds; standard deviations $<0.15$). Lower is better; \textbf{bold} marks
the large gains from DSGE-generated coverage.}
\label{tab:coverage}
\begin{tabular}{lcccc}
\toprule
& \multicolumn{2}{c}{Tail (5$\sigma$)} & \multicolumn{2}{c}{Policy regime} \\
\cmidrule(lr){2-3}\cmidrule(lr){4-5}
Architecture & narrow & wide & narrow & wide \\
\midrule
Linear      & $0.254$ & $0.407$ & $0.690$ & $\mathbf{0.021}$ \\
MLP         & $0.379$ & $0.393$ & $1.177$ & $\mathbf{0.064}$ \\
LSTM        & $3.436$ & $\mathbf{1.404}$ & $2.778$ & $\mathbf{0.109}$ \\
Transformer & $3.824$ & $\mathbf{1.498}$ & $2.913$ & $\mathbf{0.297}$ \\
NextLat     & $3.312$ & $\mathbf{2.002}$ & $2.393$ & $\mathbf{0.276}$ \\
\bottomrule
\end{tabular}
\end{table}

\begin{figure}[t]
\centering
\includegraphics[width=\linewidth]{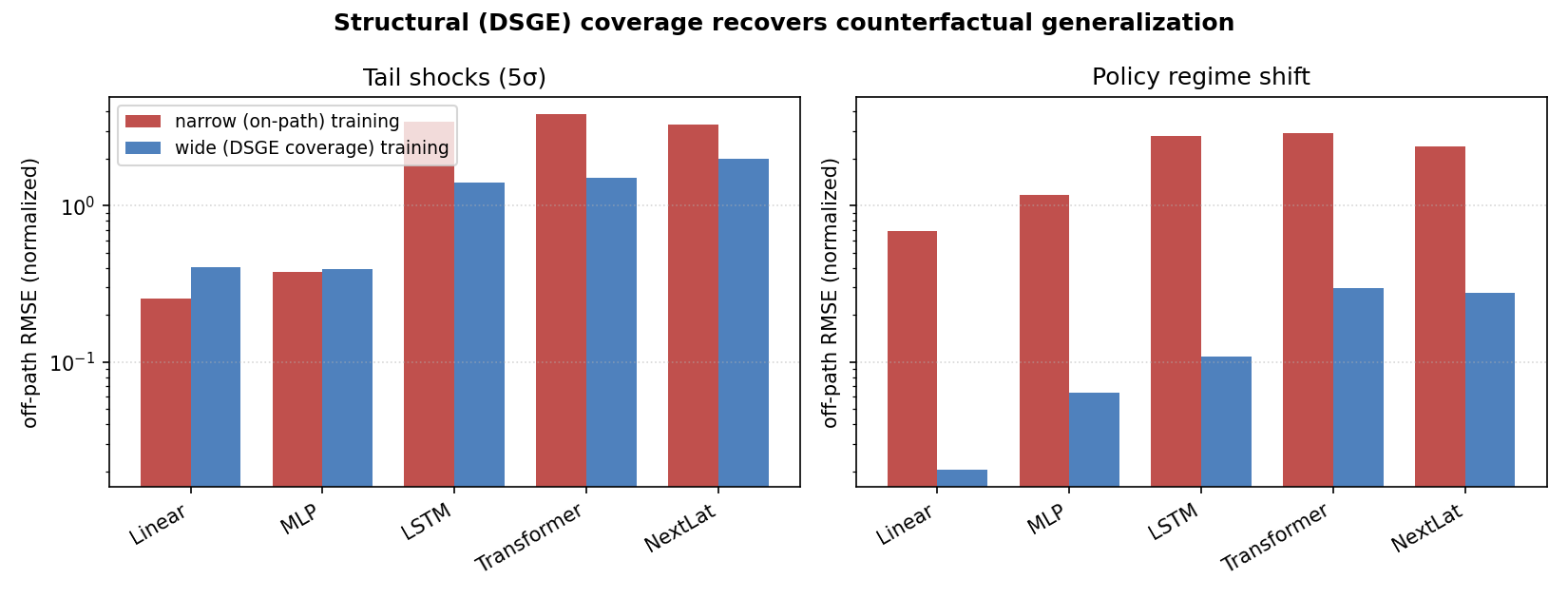}
\caption{DSGE-generated coverage recovers counterfactual generalization. The same
architectures are trained on narrow (on-path) vs.\ wide (DSGE-coverage) data and
evaluated off-path; log scale. The regime test (right) is unobservable in real data and
only a structural model can synthesize it.}
\label{fig:coverage}
\end{figure}

The effect is dramatic on the \textbf{policy-regime} test---a 10--40$\times$ reduction in
RMSE---precisely the setting that is unobservable in real data and that only a structural
model can synthesize. On \textbf{tail} events the flexible black-box models improve
${\sim}2.5\times$, because they have the capacity to exploit the added coverage. The
mechanism is exactly the coverage principle of \citet{schaab2026ewm}, here used not to
solve the model but to \emph{teach a world model where data cannot go}. This is the
practical content of our thesis: structure helps off-path because a structural model is
a counterfactual data generator.

\paragraph{Trade-off.}
Wide-coverage training spreads capacity over a larger region, so on-path sharpness drops
(wide-trained sequence models have higher on-path RMSE). The net effect is a world model
that is less precise on-path but far more robust off-path---usually the right trade for
planning and policy evaluation, where the states that matter are exactly the ones absent
from historical data.

\subsection{The pattern replicates on NK, TANK, and TCM}
\label{sec:replication}
We repeat the identical protocol on three richer economies whose policy ``action'' is
different in kind: NK (monetary policy, with nominal rigidity and a Taylor rule), TANK
(fiscal policy, with two household types, public debt, and a failure of Ricardian
equivalence), and TCM (a two-country open economy, with a real exchange rate, terms of
trade, and international bond positions; 45 variables, 6 shocks).
Table~\ref{tab:multimodel} and Figures~\ref{fig:coverage_nktank}--\ref{fig:coverage_tcm}
report the result. \textbf{H1 holds for the tail split throughout:}
trained on-path, the flexible sequence/attention models again blow up under $5\sigma$
tail shocks (RMSE ${\approx}3.2$--$4.3$). On the \emph{regime} split, by contrast, the
narrow-trained error is already \emph{small} here (RMSE ${\approx}0.16$--$0.28$): the
counterfactual policies (a hawkish Taylor rule; aggressive fiscal stabilization; an
exchange-rate-leaning rule) \emph{contract} the ergodic region rather than shift its
center, so even narrow-trained models stay near the data they saw and there is little
collapse to recover. \textbf{H2 is accordingly clearest on the tail,} where DSGE-generated
wide coverage---same sample size---roughly halves the error for the expressive models. On
the regime split, wide coverage reliably helps the lower-capacity learners (linear/MLP
improve $3$--$8\times$), but its benefit shrinks and, for the most expressive models in the
contracting-support environments, can vanish or mildly reverse (e.g.\ TANK and TCM NextLat
and TCM Transformer in Table~\ref{tab:multimodel} move by $\le0.07$ or worsen). We
therefore do \emph{not} claim a uniform regime gain: the dramatic $10$--$280\times$ regime
recovery appears wherever the counterfactual rule \emph{shifts} the ergodic support---not in
these contracting-support economies, but in RBC and in the three added support-shifting
mechanisms (DMP, E-NK, Firm; \S\ref{sec:morenv})---and we state this scope rather than
average it away. What \emph{does} survive unchanged from an
11-variable closed economy to a 45-variable two-country open economy is the tail collapse
and its coverage-based recovery---evidence it is a property of the
\emph{structure-vs-learning} contrast, not of any one model.

\begin{table}[t]
\centering
\caption{Replication on NK, TANK, and TCM: off-path normalized RMSE, narrow (on-path)
vs.\ wide (DSGE-coverage) training, fixed sample size (mean over 3 seeds). The
\emph{tail} H1$\to$H2 pattern of Table~\ref{tab:coverage} recurs under monetary (NK),
fiscal (TANK), and open-economy (TCM) actions; the \emph{regime} split shows little
narrow-trained collapse here (the counterfactual contracts rather than shifts the ergodic
support), so regime gains are small and capacity-dependent; the large regime recovery
instead appears in the support-\emph{shifting} economies (RBC, and DMP/E-NK in
Table~\ref{tab:morenv}). \textbf{Bold} marks coverage gains; note the unbolded regime
cells where wide training ties or worsens the expressive models.}
\label{tab:multimodel}
\begin{tabular}{llcccc}
\toprule
& & \multicolumn{2}{c}{Tail (5$\sigma$)} & \multicolumn{2}{c}{Policy regime} \\
\cmidrule(lr){3-4}\cmidrule(lr){5-6}
Env & Architecture & narrow & wide & narrow & wide \\
\midrule
\multirow{5}{*}{NK}
  & Linear      & $1.966$ & $2.236$ & $0.210$ & $\mathbf{0.055}$ \\
  & MLP         & $1.602$ & $1.802$ & $0.193$ & $\mathbf{0.098}$ \\
  & LSTM        & $3.715$ & $\mathbf{1.853}$ & $0.198$ & $\mathbf{0.097}$ \\
  & Transformer & $4.330$ & $\mathbf{2.009}$ & $0.194$ & $\mathbf{0.140}$ \\
  & NextLat     & $3.634$ & $\mathbf{2.022}$ & $0.198$ & $\mathbf{0.143}$ \\
\midrule
\multirow{5}{*}{TANK}
  & Linear      & $0.622$ & $1.758$ & $0.249$ & $\mathbf{0.031}$ \\
  & MLP         & $0.580$ & $1.799$ & $0.274$ & $\mathbf{0.082}$ \\
  & LSTM        & $3.289$ & $\mathbf{1.656}$ & $0.263$ & $\mathbf{0.095}$ \\
  & Transformer & $3.798$ & $\mathbf{1.643}$ & $0.280$ & $\mathbf{0.175}$ \\
  & NextLat     & $3.227$ & $\mathbf{1.836}$ & $0.276$ & $0.273$ \\
\midrule
\multirow{5}{*}{TCM}
  & Linear      & $1.640$ & $1.703$ & $0.271$ & $\mathbf{0.066}$ \\
  & MLP         & $1.328$ & $\mathbf{1.227}$ & $0.155$ & $\mathbf{0.136}$ \\
  & LSTM        & $3.649$ & $\mathbf{1.912}$ & $0.167$ & $\mathbf{0.137}$ \\
  & Transformer & $4.204$ & $\mathbf{1.839}$ & $0.179$ & $0.180$ \\
  & NextLat     & $3.752$ & $\mathbf{1.997}$ & $0.179$ & $0.243$ \\
\bottomrule
\end{tabular}
\end{table}

\begin{figure}[t]
\centering
\includegraphics[width=0.92\linewidth]{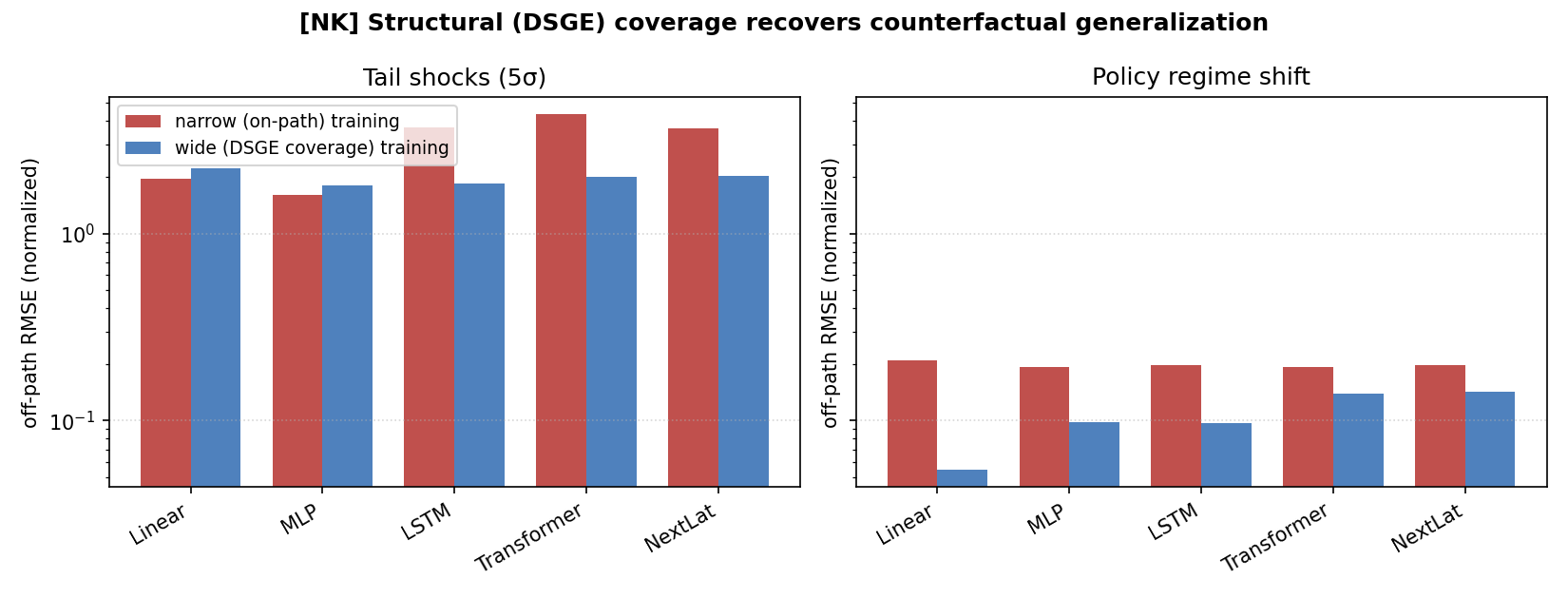}\\[2pt]
\includegraphics[width=0.92\linewidth]{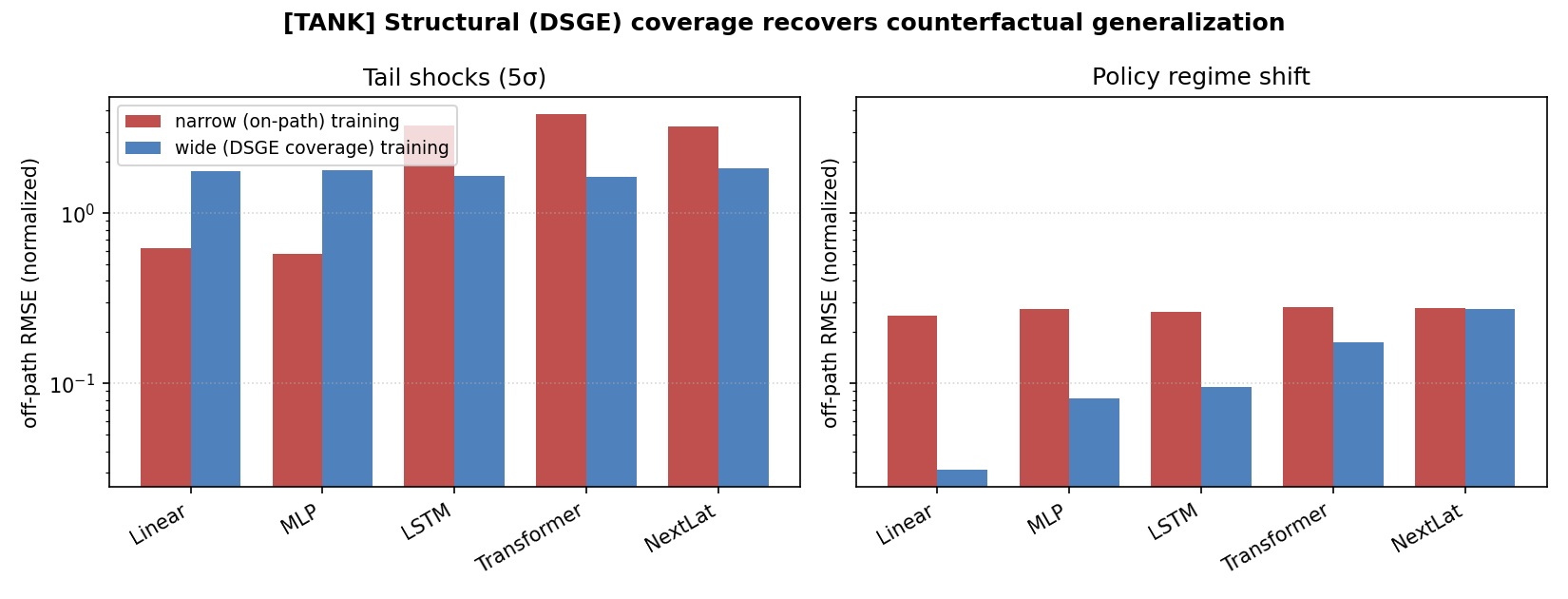}
\caption{Coverage replication on NK (top) and TANK (bottom). Same architectures trained
on narrow (on-path) vs.\ wide (DSGE-coverage) data and evaluated off-path; log scale.
As in RBC (Figure~\ref{fig:coverage}), DSGE-generated coverage lowers off-path error for
the expressive sequence/attention models, under both monetary (NK) and fiscal (TANK)
policy actions.}
\label{fig:coverage_nktank}
\end{figure}

\begin{figure}[t]
\centering
\includegraphics[width=0.92\linewidth]{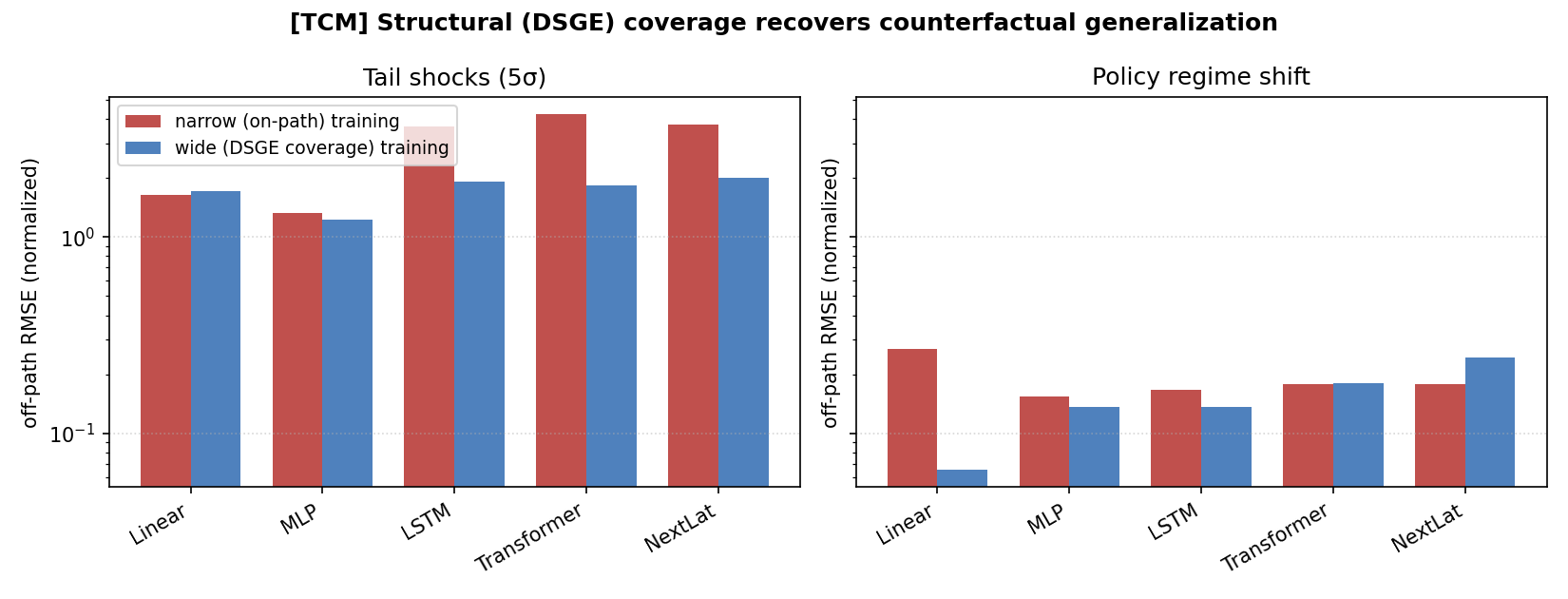}
\caption{Coverage replication on the two-country open economy (TCM; 45 variables, 6
shocks). The same H1$\to$H2 pattern holds with an open-economy monetary action
(exchange-rate-leaning Taylor rules), confirming the effect is not an artifact of model
size or closure.}
\label{fig:coverage_tcm}
\end{figure}

\subsection{Scaling to a production model: NAWM (230 variables)}
\label{sec:nawm}
To test whether the benchmark---and the result---survive at the scale of a real
policy model, we add the \textbf{New Area-Wide Model} of the European Central Bank
\citep{coenen2008nawm}, an estimated two-region (Euro Area + United States) DSGE with
\textbf{230 variables and 21 structural shocks}, taken unmodified from the
MacroModelling.jl model library. At this size a third-order solution is infeasible, so
the data are generated from the \emph{first-order} solution; the off-path \textbf{regime}
makes Euro-Area monetary policy hawkish ($\phi^{EA}_\pi{:}2{\to}4$). Results are in
Table~\ref{tab:nawm} and Figure~\ref{fig:coverage_nawm}.

Two findings stand out. First, \textbf{H1 survives to production scale}: trained on-path,
the sequence/attention models again collapse under $5\sigma$ tail shocks (RMSE
${\approx}3.4$--$3.5$) even though the data-generating process is \emph{linear}---they
overfit the ergodic region and extrapolate badly---while the linear ridge, which matches
the DGP, is essentially exact. Second, \textbf{H2 holds on the tail but interacts with
capacity}: DSGE-generated coverage cuts the sequence models' tail error by ${\sim}40\%$
($3.4{\to}2.1$, $3.5{\to}1.9$), but on the (already small) regime test only the
lower-capacity learners exploit the added coverage---the linear model improves
$7\times$ ($0.17{\to}0.025$)---whereas the fixed-size sequence models, underfit on a
230-dimensional target, spread their capacity too thin and do not. The headline collapse
and its coverage-based tail recovery are thus robust from an 11-variable RBC to a
230-variable estimated central-bank model; the regime result additionally exposes that
exploiting structural coverage requires sufficient model capacity---itself a useful signal
the benchmark surfaces.

\begin{table}[t]
\centering
\caption{NAWM (230 variables, 21 shocks; first-order DGP): off-path normalized RMSE,
narrow vs.\ wide training (mean over 3 seeds). H1 and the tail-coverage gain (H2) survive
to production scale; on the small regime test, only lower-capacity learners exploit
coverage.}
\label{tab:nawm}
\begin{tabular}{lcccc}
\toprule
& \multicolumn{2}{c}{Tail (5$\sigma$)} & \multicolumn{2}{c}{Policy regime} \\
\cmidrule(lr){2-3}\cmidrule(lr){4-5}
Architecture & narrow & wide & narrow & wide \\
\midrule
Linear      & $0.331$ & $0.331$ & $0.174$ & $\mathbf{0.025}$ \\
MLP         & $0.858$ & $\mathbf{0.498}$ & $0.092$ & $\mathbf{0.084}$ \\
LSTM        & $3.375$ & $\mathbf{2.055}$ & $0.172$ & $0.300$ \\
Transformer & $3.526$ & $\mathbf{1.860}$ & $0.156$ & $0.251$ \\
NextLat     & $3.480$ & $\mathbf{2.129}$ & $0.192$ & $0.418$ \\
\bottomrule
\end{tabular}
\end{table}

\begin{figure}[t]
\centering
\includegraphics[width=0.92\linewidth]{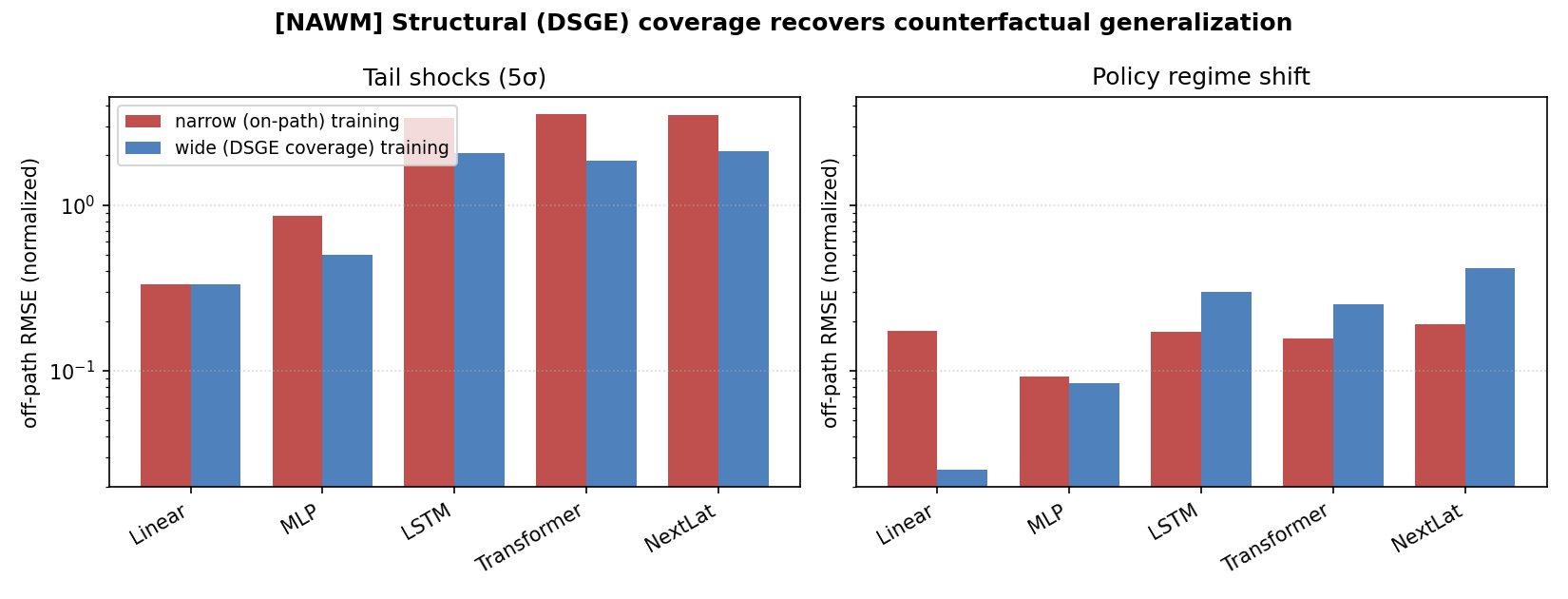}
\caption{Scaling to the ECB New Area-Wide Model (NAWM; 230 variables, 21 shocks). The
off-path collapse (H1) and its tail recovery from DSGE coverage (H2) persist at the scale
of a production policy model; the linear baseline matches the first-order DGP and is exact
on the tail.}
\label{fig:coverage_nawm}
\end{figure}

\subsection{Three further mechanisms: search unemployment, climate, and firm entry
(DMP, E-NK, Firm)}
\label{sec:morenv}
To show the protocol extends beyond the real/monetary/fiscal/open-economy block and to test
whether the dramatic \emph{regime} recovery is special to RBC, we add three environments
whose mechanisms appear in none of the others. \textbf{DMP} is a
Diamond--Mortensen--Pissarides search-and-matching economy (12 variables, 1 shock;
Appendix~\ref{app:dmp}) with an explicit labor market---vacancies, a matching function, and
a Nash-bargained wage---whose ``action'' is the unemployment-benefit / value-of-non-work
$b$; the off-path regime raises $b$ (a more generous UI rule). \textbf{E-NK} is the NK
economy augmented with a climate block (17 variables, 3 shocks; Appendix~\ref{app:enk}):
emissions proportional to output, an atmospheric carbon \emph{stock}, and a
carbon$\to$productivity \emph{damage} feedback; its action is climate policy, the carbon tax
that sets abatement intensity $\bar\mu$ ($\bar\mu{:}0{\to}0.4$). \textbf{Firm} is an
endogenous-entry economy (Bilbiie--Ghironi--Melitz; 11 variables, 1 shock;
Appendix~\ref{app:firm}) with a mass of firms, entrants, and firm value on a product-variety
margin; its action is entry policy---the sunk entry cost $f_E$, which the off-path regime
lowers ($f_E{:}1.0{\to}0.7$, deregulation). All three regimes \emph{shift} the ergodic
support (more generous UI raises steady-state unemployment; tighter abatement lowers the
carbon stock; lower entry cost raises the firm count from ${\approx}8.7$ to ${\approx}12.4$),
the configuration under which Proposition~\ref{prop:lucas} predicts the largest gap.
Table~\ref{tab:morenv} reports the result.

\textbf{H1} recurs: trained on-path, the sequence/attention models again blow up under
$5\sigma$ tail shocks (RMSE ${\approx}3.1$--$4.4$). On the \emph{regime} split all three
economies show a \emph{large} narrow-trained collapse (RMSE ${\approx}3.3$--$18.9$), exactly
as in RBC and unlike the support-\emph{contracting} NK/TANK/TCM counterfactuals---direct
evidence that the size of the regime gap is governed by whether the counterfactual moves the
support, not by the model family. \textbf{H2} then recovers the regime split dramatically
($10$--$670\times$; wide RMSE ${\approx}0.02$--$0.72$), confirming that the RBC-scale regime
recovery is a property of \emph{support-shifting} counterfactuals, now replicated under
labor-market, climate, and firm-entry actions. On the \emph{tail} split the picture is the
familiar mixed one: coverage roughly halves DMP's sequence-model error, but for E-NK and Firm
---whose slow, near-unit-root stocks (carbon stock; firm mass) make the wide regime segments
sit far from the on-path normalizer---wide training spreads capacity and does \emph{not} help
the tail, an instance of the coverage$\times$capacity interaction we flag throughout (it does
not undo the regime gain, which is what the structural generator uniquely supplies).

\begin{table}[t]
\centering
\caption{Three further mechanisms: off-path normalized RMSE, narrow (on-path) vs.\ wide
(DSGE-coverage) training, fixed sample size (mean over 3 seeds). All three regimes
\emph{shift} the ergodic support, and---as in RBC---wide coverage recovers the regime split
by $10$--$670\times$. \textbf{Bold} marks coverage gains on the expressive models; note the
tail cells where a slow stock variable makes wide coverage spread capacity and not help.}
\label{tab:morenv}
\begin{tabular}{llcccc}
\toprule
& & \multicolumn{2}{c}{Tail (5$\sigma$)} & \multicolumn{2}{c}{Policy regime} \\
\cmidrule(lr){3-4}\cmidrule(lr){5-6}
Env & Architecture & narrow & wide & narrow & wide \\
\midrule
\multirow{5}{*}{DMP}
  & Linear      & $0.219$ & $0.431$ & $5.681$ & $\mathbf{0.020}$ \\
  & MLP         & $0.446$ & $0.504$ & $3.345$ & $\mathbf{0.036}$ \\
  & LSTM        & $3.149$ & $\mathbf{2.160}$ & $3.874$ & $\mathbf{0.154}$ \\
  & Transformer & $3.678$ & $\mathbf{2.875}$ & $4.230$ & $\mathbf{0.146}$ \\
  & NextLat     & $3.074$ & $\mathbf{2.920}$ & $3.672$ & $\mathbf{0.357}$ \\
\midrule
\multirow{5}{*}{E-NK}
  & Linear      & $1.952$ & $1.862$ & $6.694$ & $\mathbf{0.114}$ \\
  & MLP         & $1.546$ & $3.990$ & $4.560$ & $\mathbf{0.089}$ \\
  & LSTM        & $3.703$ & $4.119$ & $6.910$ & $\mathbf{0.153}$ \\
  & Transformer & $4.404$ & $6.522$ & $6.912$ & $\mathbf{0.308}$ \\
  & NextLat     & $3.738$ & $6.372$ & $6.641$ & $\mathbf{0.222}$ \\
\midrule
\multirow{5}{*}{Firm}
  & Linear      & $0.192$ & $0.361$ & $18.882$ & $\mathbf{0.028}$ \\
  & MLP         & $0.326$ & $1.897$ & $12.899$ & $\mathbf{0.096}$ \\
  & LSTM        & $3.458$ & $5.276$ & $14.956$ & $\mathbf{0.266}$ \\
  & Transformer & $3.928$ & $5.209$ & $15.114$ & $\mathbf{0.297}$ \\
  & NextLat     & $3.204$ & $10.565$ & $14.647$ & $\mathbf{0.716}$ \\
\bottomrule
\end{tabular}
\end{table}

\section{Discussion and Limitations}

\paragraph{What the results mean.}
Our claim is deliberately scoped: a DSGE is a \emph{structured world model}, and its
structure---through coverage and constraint satisfaction---buys the counterfactual
generalization that learned world models lack. The experiments make this concrete and
measurable. The collapse of flexible learners off-path (H1) is not a tuning artifact: it
reproduces across eight economies, two solution orders (first- and third-order, i.e.\ both
linear and genuinely nonlinear data-generating processes), and model sizes
spanning $11$ to $230$ variables, and it \emph{worsens with capacity}---the models with
the most expressive learned dynamics fail hardest exactly where data are absent. The
recovery (H2) localizes \emph{why} structure helps to a single operational mechanism:
not the constraints per se, but the ability to \emph{generate} coverage of states history
never visited. A structural model is, first and foremost, a counterfactual data
generator; that is the property a purely learned world model cannot replicate, because it
has no access to states outside its training distribution.

\paragraph{A lesson beyond economics.}
The finding is an instance of a general principle for model-based agents: when the
decision-relevant states lie off the data manifold---policy counterfactuals, rare events,
sim-to-real gaps---a generative structural prior dominates a data-bounded learner, and the
cleanest way to inject that prior is as \emph{synthetic coverage} rather than as a soft
penalty. This reframes the role of mechanistic models (physics engines, simulators,
economic theory) in the world-model program: their value is not only as ground truth but
as \emph{distribution shapers} that teach a learned model where reality has not yet gone.
Economics is an unusually sharp testbed for this idea because its counterfactuals
(policy-rule changes) are both first-class objects of theory and provably absent from any
single historical sample---the Lucas critique made empirical. We stress the symmetric
danger: because the off-path region is exactly where model classes disagree most, synthetic
coverage is only as trustworthy as the generating model \emph{there}, and using it amounts
to trusting structure precisely where data cannot check it. That is not a reason to avoid
structural coverage but the strongest argument for \emph{learning} structure that data can
later correct (Paper~2), rather than imposing a fixed model wholesale.

\paragraph{When does coverage help? The capacity interaction.}
The NAWM scale test (\S\ref{sec:nawm}) refines the headline. Coverage reliably recovers
the \emph{tail} gap at every scale, but exploiting coverage on a counterfactual-regime
test requires sufficient model capacity: on the $230$-variable target the fixed-size
sequence models are underfit and cannot absorb the broader training distribution, so their
regime error does not improve (and can degrade), whereas the lower-capacity linear and
MLP learners improve sharply. Two readings follow. First, the benchmark surfaces a real
and under-appreciated trade-off---coverage and capacity must scale together---rather than
hiding it. Second, it suggests that the payoff from structural coverage is largest when
paired with a model expressive enough to use it, which is precisely the learned
structured world model we turn to next.

\paragraph{Limitations.}
\begin{enumerate}
\item \textbf{Solution order at scale, and what the scale test does (not) show.} The four
main environments use a pruned third-order solution (genuine nonlinearity); NAWM, at $230$
variables, is solved only at first order, so its tail is \emph{linear} and a linear baseline
matches the DGP exactly. The NAWM result therefore validates the harness and isolates pure
out-of-support extrapolation of the learners, but it does \emph{not} test the
\emph{nonlinear} off-path claim at production scale---that claim is established only at the
11--45-variable scale. A second-order solve of NAWM is the natural next step and we flag it
as the key missing piece of the scale story.

\item \textbf{Structure vs.\ coverage, and baseline strength.} Our design shows
DSGE-generated coverage helps, but does not yet isolate \emph{structure} from coverage
\emph{per se}. The decisive control is a non-structural coverage source---regime-dummy VARs,
parameter-perturbation augmentation, or a learned simulator---matched on sample size and
support; only then does ``the equilibrium structure matters'' separate from ``a wider
training region matters.'' Relatedly, baselines are shared and modestly tuned for protocol
uniformity, and standard off-path-robustness interventions (input-noise augmentation,
weight/spectral norms, distributional/quantile heads, ensembles) are not exhausted; H1 is
thus the behavior of standard, comparably-tuned learners, not a proof of an intrinsic
ceiling. We regard the non-structural control as the single most important addition for the
next revision and Paper~2.
\item \textbf{Structure is imposed, not learned.} The structured references here---the
DSGE itself, its linearized oracle, and the static-equilibrium penalty---are
\emph{exact-or-fixed} solutions, not \emph{learned} structured models. The negative
results (the first-order oracle is a poor off-path world model; static penalties do not
help) show that imposing approximate structure is insufficient; learning it is the central
object of Paper~2.
\item \textbf{Misspecification.} DSGE-generated coverage is only as good as the model
class. If the true economy departs from the DSGE, synthetic off-path data can be confidently
wrong---a real risk, and itself an argument for \emph{learning} structure (which can adapt)
over imposing it wholesale.
\item \textbf{Scope of the task.} We study one-step prediction (T1); compounding error in
multi-step rollouts (T2) and planning regret (T3) are released with the benchmark but not
analyzed here, as is the belief-state sufficiency probe. We also assume full observability
($z_t$ contains the structural state); the partially observed case, where the belief state
is a posterior, is left to future work. Baselines are deliberately modestly tuned and
shared across environments to keep the protocol uniform rather than to maximize any single
number.
\end{enumerate}
We regard limitations~2 and~3 as the central open problems, and the bridge to a learned,
NextLat-style structured world model for economics.

\section{Conclusion}
We recast DSGE models as \emph{structured world models} and showed, via
Proposition~\ref{prop:bridge}, that their state variable is the \emph{belief state} that
modern transformer world models are trained to learn---making two literatures that grew up
apart provably about the same object. To study the consequence we built \textbf{DSGE-Gym},
a benchmark that draws train and test from the same solved model but different regions of
its state space, turning off-path generalization into a clean measurement. Across eight
economies---from an $11$-variable RBC to the $230$-variable New Area-Wide Model of the
ECB, and spanning real, monetary, fiscal, open-economy, labor-market, climate, and
firm-entry policy actions---the story is consistent: purely learned world models
match the structured dynamics on-path but collapse off-path, and \emph{re-emerge} when the
\emph{same} architectures are trained on data the DSGE generates across rare and
counterfactual states. Structure helps off-path not as a constraint but as a generator of
coverage.

The negative results sharpen the agenda. A first-order oracle and static-equilibrium
penalties---approximate, imposed structure---do not close the gap, and at production scale
coverage pays off only when model capacity can absorb it. Both point the same way:
\emph{learn} a compact latent world model that carries DSGE structure as a training
signal rather than as a fixed scaffold. That is Paper~2 of this agenda (a NextLat-style
structured latent world model for economics); Paper~3 builds multi-agent digital-twin
economies on top of it. DSGE-Gym, its protocol, and all code are released so the community
can measure---and contest---counterfactual generalization of world models on common
ground.

\section*{Reproducibility and Resource Availability}
All code (Julia DSGE environments built on MacroModelling.jl and PyTorch baselines),
generated datasets, and scripts to reproduce every table and figure are released at the
project repository. DSGE-Gym is designed as a reusable benchmark: each environment
exposes the same interface and counterfactual test sets. Consistent with Datasets \&
Benchmarks norms, the artifact ships with a \emph{datasheet} documenting every environment
(variables, shocks, calibration source, steady-state verification against published
values), per-environment dataset sizes and generation seeds, the train/test split
definitions of Appendix~\ref{app:tests} (including the leave-one-regime-out and
non-structural-coverage variants discussed above), and an OSI-approved open-source license;
an anonymized artifact link is provided to reviewers and de-anonymized for camera-ready.

\bibliographystyle{plainnat}
\bibliography{references}

\newpage
\appendix

\section{DSGE-Gym environment specifications}
\label{app:models}

This appendix gives the complete equilibrium system, calibration, and steady state of
each environment reported in the body (RBC, NK, TANK, TCM), together with the explicit
world-model mapping. Notation follows \citet{nispilandi2018rbc,nispilandi2021nk}: $c$ is
consumption, $h$ labor, $k$ capital, $i$ investment, $y$ output, $w$ the real wage, $rk$
the rental rate of capital, $rr$ the real interest rate, $\lambda$ the marginal utility
of consumption, $a$ total factor productivity, and $g$ government spending; in the
nominal models $r$ is the gross nominal rate, $\pi$ gross inflation, $mc$ real marginal
cost, and $q$ Tobin's $q$. Bracketed time indices in the verified source
(\texttt{MacroModelling.jl}) map to the standard convention $x[-1]\!\equiv\!x_{t-1}$,
$x[0]\!\equiv\!x_t$, $x[1]\!\equiv\!E_t x_{t+1}$. Every steady state below is solved
numerically and matches the published calibration to all reported digits.

\paragraph{World-model mapping.}
Each environment is cast as $z_{t+1}=F(z_t,a_t,\varepsilon_{t+1})$ with the assignment in
Table~\ref{tab:mapping}. The \emph{state} $z_t$ is the full vector of model variables
(the predetermined/exogenous components---capital, lagged investment, debt, the lagged
policy rate, and the exogenous processes---are the minimal Markov state; the remaining
jump variables are functions of these). The \emph{action} $a_t$ is the policy block whose
counterfactual is studied; the \emph{shocks} $\varepsilon_t$ are the structural
innovations.

\begin{table}[h]
\centering
\caption{World-model assignment per environment. ``Min.\ state'' lists the predetermined
and exogenous variables that constitute the Markov state (jumpers omitted).}
\label{tab:mapping}
\small
\begin{tabular}{llll}
\toprule
Env & Min.\ state (predetermined + exog.) & Action (policy) & Shocks $\varepsilon_t$ \\
\midrule
RBC  & $k_{t-1},a_t,g_t$                       & gov.\ spending $g$ (rule)        & $\varepsilon^a,\varepsilon^g$ \\
NK   & $k_{t-1},i_{t-1},r_{t-1},a_t,g_t$       & nominal rate $r$ (Taylor rule)  & $\varepsilon^a,\varepsilon^g,\varepsilon^m$ \\
TANK & $k_{t-1},i_{t-1},d_{t-1},r_{t-1},a_t,g_t$ & fiscal financing (tax rules)  & $\varepsilon^a,\varepsilon^g,\varepsilon^m$ \\
TCM  & {\scriptsize $k,i,r,a,g,b_H,b_F,s,p_H$ (\& Foreign)} & two Taylor rules ($\phi_e$) & {\scriptsize $\varepsilon^a,\varepsilon^g,\varepsilon^m$ (\& *)} \\
\bottomrule
\end{tabular}
\end{table}

\subsection{RBC}
\label{app:rbc}
The real-business-cycle environment \citep{nispilandi2018rbc} has 11 endogenous variables
$\{\lambda,c,rr,rk,w,h,y,k,i,g,a\}$ and 2 shocks $\{\varepsilon^a,\varepsilon^g\}$. With
investment adjustment cost set to zero, Tobin's $q\equiv1$ is dropped. The equilibrium
conditions are
\begin{align}
\lambda_t &= c_t^{-\sigma}, \qquad \kappa_L h_t^{\phi} = \lambda_t w_t, \\
\lambda_t &= \beta\,\lambda_{t+1}\,rr_t, \qquad
1 = \beta\,\tfrac{\lambda_{t+1}}{\lambda_t}\,(rk_{t+1}+1-\delta), \\
k_t &= (1-\delta)k_{t-1}+i_t, \qquad y_t = a_t k_{t-1}^{\alpha} h_t^{1-\alpha}, \\
\alpha\,y_t &= rk_t k_{t-1}, \qquad (1-\alpha)\,y_t = w_t h_t, \qquad
y_t = c_t + i_t + g_t, \\
\log a_t &= (1-\rho_a)\log\bar a + \rho_a\log a_{t-1} + \sigma_a\varepsilon^a_t, \\
\log g_t &= (1-\rho_g)\log\bar g + \rho_g\log g_{t-1} + \sigma_g\varepsilon^g_t.
\end{align}
Calibration (quarterly): $\beta{=}0.99$, $\alpha{=}0.33$, $\delta{=}0.025$,
$\sigma{=}2$, $\phi{=}1$, $\rho_a{=}\rho_g{=}0.9$, $\bar g{=}0.2$,
$\sigma_a{=}\sigma_g{=}0.01$, with $\bar a{=}0.98377$ and $\kappa_L{=}18.889$ calibrated
so that $y\!\approx\!1$, $h\!\approx\!1/3$. The off-path \textbf{regime} split raises the
spending process to $(\bar g,\rho_g){=}(0.25,0.97)$; the process remains stationary
($\rho_g<1$), so the ergodic set is well defined and the counterfactual shifts its
\emph{location} (not just its shape). This relocation of the support is why RBC exhibits the
largest regime collapse---and the largest coverage recovery---in the suite
(\S\ref{sec:replication}).

\subsection{New Keynesian (NK)}
\label{app:nk}
The NK environment \citep{nispilandi2021nk} adds monopolistic competition, Rotemberg
price adjustment, investment adjustment costs, and a Taylor rule: 15 variables
$\{\lambda,c,rr,rk,w,h,y,k,q,i,r,mc,\pi,g,a\}$ and 3 shocks
$\{\varepsilon^a,\varepsilon^g,\varepsilon^m\}$. The equilibrium conditions are
\begin{align}
\lambda_t &= c_t^{-\sigma}, \qquad
1 = \beta\,\tfrac{\lambda_{t+1}}{\lambda_t}\,\tfrac{r_t}{\pi_{t+1}}, \qquad
\kappa_L h_t^{\phi} = \lambda_t w_t, \\
1 &= \beta\,\tfrac{\lambda_{t+1}}{\lambda_t}\,\tfrac{rk_{t+1}+(1-\delta)q_{t+1}}{q_t}, \\
k_t &= (1-\delta)k_{t-1} + \big[1-\tfrac{\kappa_I}{2}(i_t/i_{t-1}-1)^2\big]i_t, \\
1 &= q_t\big[1-\tfrac{\kappa_I}{2}(i_t/i_{t-1}-1)^2-\kappa_I\tfrac{i_t}{i_{t-1}}(i_t/i_{t-1}-1)\big] \nonumber\\
  &\quad + \kappa_I\beta\,\tfrac{\lambda_{t+1}}{\lambda_t}q_{t+1}\big[(i_{t+1}/i_t)^2(i_{t+1}/i_t-1)\big], \\
rr_t &= \tfrac{r_t}{\pi_{t+1}}, \qquad
y_t = a_t k_{t-1}^{\alpha} h_t^{1-\alpha}, \\
(1-\alpha)mc_t y_t &= w_t h_t, \qquad \alpha\,mc_t y_t = rk_t k_{t-1}, \\
\pi_t(\pi_t-\bar\pi) &= \beta\,\tfrac{\lambda_{t+1}}{\lambda_t}\,\pi_{t+1}(\pi_{t+1}-\bar\pi)\tfrac{y_{t+1}}{y_t}
   + \tfrac{\varepsilon}{\kappa_P}\big(mc_t-\tfrac{\varepsilon-1}{\varepsilon}\big), \\
y_t &= c_t + i_t + g_t + \tfrac{\kappa_P}{2}(\pi_t-\bar\pi)^2 y_t, \\
\tfrac{r_t}{\bar r} &= \big(\tfrac{r_{t-1}}{\bar r}\big)^{\rho_r}
   \big[(\pi_t/\bar\pi)^{\phi_\pi}(y_t/\bar y)^{\phi_y}\big]^{1-\rho_r}\exp(\sigma_m\varepsilon^m_t), \\
\log a_t &= (1-\rho_a)\log\bar a + \rho_a\log a_{t-1} + \sigma_a\varepsilon^a_t, \\
\log g_t &= (1-\rho_g)\log\bar g + \rho_g\log g_{t-1} + \sigma_g\varepsilon^g_t.
\end{align}
Here $\varepsilon$ is the elasticity of substitution across varieties and $\kappa_P$ the
Rotemberg cost. Calibration: $\beta{=}0.99$, $\alpha{=}0.33$, $\varepsilon{=}6$,
$\delta{=}0.025$, $\sigma{=}2$, $\phi{=}1$, $\bar g{=}0.2$, $\bar\pi{=}\bar y{=}1$,
$\bar r{=}1/\beta$, $\phi_\pi{=}1.5$, $\phi_y{=}0.125$, $\kappa_I{=}2.48$,
$\rho_a{=}\rho_g{=}0.9$, $\rho_r{=}0.8$, $\kappa_P{=}28.003$ (Calvo-equivalent $0.66$),
$\sigma_a{=}\sigma_g{=}0.01$, $\sigma_m{=}0.0025$, with $\bar a{=}1.0584$ and
$\kappa_L{=}13.768$. Steady state: $\pi{=}1$, $y{=}1$, $h{=}1/3$, $mc{=}5/6$,
$r{=}rr{=}1/\beta$, $q{=}1$, $k{=}7.835$. The \textbf{regime} split sets a hawkish
Taylor rule $(\phi_\pi,\phi_y){=}(3.0,0.5)$.

\subsection{Two-agent New Keynesian (TANK)}
\label{app:tank}
TANK \citep{galilopezsalidovalles2007} splits households into a fraction $1-\eta$ of
optimizing (Ricardian) agents and $\eta$ rule-of-thumb (hand-to-mouth) agents, adds
public debt $d$ and two tax-financing rules: 22 variables
$\{\lambda,c,rr,rk,w,h,y,k,q,i,r,mc,\pi,g,a,c^o,c^r,h^o,h^r,t^r,t^o,d\}$ and 3 shocks.
Superscripts $o,r$ denote optimizers and rule-of-thumb agents. The Euler, capital,
Phillips, production, Taylor, and exogenous-process equations are as in NK (with
optimizer marginal utility $\lambda_t=(c^o_t)^{-\sigma}$ and labor supply
$\kappa_L (h^o_t)^{\phi}=\lambda_t w_t$); the additional blocks are
\begin{align}
\kappa_L (h^r_t)^{\phi}(c^r_t)^{\sigma} &= w_t, \qquad
c_t = (1-\eta)c^o_t + \eta c^r_t, \qquad
h_t = (1-\eta)h^o_t + \eta h^r_t, \\
c^r_t &= w_t h^r_t - t^r_t, \qquad
g_t + \tfrac{r_{t-1}}{\pi_t}d_{t-1} = (1-\eta)t^o_t + \eta t^r_t + d_t, \\
t^r_t - \bar t^r &= \phi_d(d_{t-1}-\bar d) + \phi_g(g_{t-1}-\bar g), \\
t^o_t - \bar t^o &= \phi_d(d_{t-1}-\bar d) + \phi_g(g_{t-1}-\bar g).
\end{align}
Because rule-of-thumb agents consume disposable income, Ricardian equivalence fails: how
spending is financed (the $\phi_d,\phi_g$ rules) affects aggregate consumption. This is
why the TANK \textbf{action} is the fiscal-financing block. Calibration:
$\beta{=}0.99$, $\alpha{=}1/3$, $\varepsilon{=}6$, $\delta{=}0.025$, $\sigma{=}1$,
$\phi{=}0.2$, $\eta{=}0.5$, $\bar g{=}0.2$, $\bar d{=}4$ (annual debt/GDP $=1$),
$\phi_\pi{=}1.5$, $\phi_y{=}0$, $\kappa_I{=}2.48$, $\rho_a{=}\rho_g{=}0.9$, $\rho_r{=}0$,
$\phi_g{=}0.1$, $\phi_d{=}0.33$, $\kappa_P{=}58.252$, with $\bar a{=}1.0438$,
$\kappa_L{=}3.448$, $\bar t^r{=}-0.0466$, $\bar t^o{=}0.5274$. Steady state:
$c{=}c^o{=}c^r{=}0.602$, $d{=}4$, $k{=}7.914$, $\lambda{=}1.661$, $y{=}1$. The
\textbf{regime} split uses aggressive stabilization $(\phi_d,\phi_g){=}(0.6,0.5)$.
\emph{Determinacy caveat:} large hand-to-mouth shares $\eta$ violate the Taylor principle
(inverted aggregate-demand logic, \citealp{bilbiie2008limited}) and render the solution
explosive; we therefore hold $\eta$ fixed and vary fiscal/persistence parameters for
coverage.

\subsection{Two-country open economy (TCM)}
\label{app:tcm}
TCM \citep[after][]{galilopezsalidovalles2007} is a two-country New Keynesian world
economy (Home of size $n$, Foreign of size $1-n$) with incomplete markets, producer-currency
pricing and the law of one price, physical capital, and Rotemberg pricing: \textbf{45
variables and 6 shocks} (Home $\{\varepsilon^a,\varepsilon^g,\varepsilon^m\}$ and Foreign
$\{\varepsilon^{a*},\varepsilon^{g*},\varepsilon^{m*}\}$). Within each country the block
mirrors the NK equations of Appendix~\ref{app:nk} (marginal utility, capital Euler,
labor supply/demand, investment FOC, capital accumulation, production, Rotemberg Phillips
curve, Taylor rule, exogenous processes), written for Home goods (Foreign vars starred).
The distinctive open-economy blocks are
\begin{align}
1 &= \beta\,\tfrac{\lambda_{t+1}}{\lambda_t}\,\tfrac{r^*_t}{\pi^*_{t+1}}\,\tfrac{s_{t+1}}{s_t}
   - \kappa_D(b_{Ft}-\bar b_F),   &&\text{(Foreign-bond Euler / UIP)} \\
\tfrac{s_t}{s_{t-1}} &= \Delta e_t\,\tfrac{\pi^*_t}{\pi_t}, \qquad
   p_{Ht}=s_t p^*_{Ht},\quad p_{Ft}=s_t p^*_{Ft}, &&\text{(RER, law of one price)} \\
1 &= (1-\gamma)p_{Ht}^{1-\eta}+\gamma p_{Ft}^{1-\eta}, &&\text{(Home CPI aggregator)} \\
n\,b_{Ht}+(1-n)b^*_{Ht} &= 0, \qquad n\,b_{Ft}+(1-n)b^*_{Ft}=0, &&\text{(bond clearing)}
\end{align}
together with Home/Foreign goods-market clearing that mixes domestic and imported demand,
the trade balance, and $gdp_t=p_{Ht}y_{Ht}$. The Taylor rules optionally respond to the
gross depreciation $\Delta e_t$ with coefficient $\phi_e$ (Home $+\phi_e$, Foreign
$-\phi_e$); $\kappa_D$ is a small bond-adjustment cost that renders net foreign assets
stationary. Calibration: $\beta{=}0.99$, $\alpha{=}0.33$, $\varepsilon{=}6$,
$\delta{=}0.025$, $\sigma{=}2$, $\phi{=}1$, $\eta{=}1.5$ (intratemporal substitution),
$n{=}0.2$, openness $\omega{=}0.3$ giving $\gamma{=}0.2444$, $\gamma^*{=}0.0556$;
$\phi_\pi{=}1.5$, $\phi_y{=}0.125$, $\phi_e{=}0$, $\kappa_I{=}2.48$, $\kappa_D{=}0.01$,
$\rho_a{=}\rho_g{=}0.9$, $\rho_r{=}0.8$, $\kappa_P{=}28.003$; external debt/GDP $D{=}0$,
$D^*{=}0.1$. The steady state is \emph{not} analytic (the source solves $p_H$ by a
numerical root-find); our solver reproduces all published values, e.g.\ $p_H{=}1.0206$,
real exchange rate $s{=}0.9443$, $gdp{=}1$, $gdp^*{=}1.1649$, $k{=}7.835$, $k^*{=}9.127$,
$b_H{=}1.76$, $b^*_H{=}-0.44$, $tb{=}-0.0178$. The \textbf{regime} split turns on the
exchange-rate response $\phi_e{:}0{\to}0.5$.

\subsection{Search-and-matching unemployment (DMP)}
\label{app:dmp}
The DMP environment embeds a Diamond--Mortensen--Pissarides labor market in an RBC shell:
11 endogenous variables plus TFP $\{a,c,n,u,v,\theta,m,q,f,w,y,\lambda\}$ and a single TFP
shock $\varepsilon^a$. Let $n$ be employment, $u=1-(1-s)n_{t-1}$ the job seekers after
separations, $v$ vacancies, $\theta=v/u$ tightness, $m=\chi u^{\xi}v^{1-\xi}$ matches,
$f=m/u$ and $q=m/v$ the job-finding and vacancy-filling rates. The equilibrium is
\begin{align}
\lambda_t &= c_t^{-\sigma}, \qquad n_t=(1-s)n_{t-1}+m_t, \qquad y_t=a_t n_t, \qquad
c_t=y_t-\kappa v_t,\\
\tfrac{\kappa}{q_t} &= \beta\,\tfrac{\lambda_{t+1}}{\lambda_t}\big(a_{t+1}-w_{t+1}
   +(1-s)\tfrac{\kappa}{q_{t+1}}\big), \qquad
w_t=\eta\,(a_t+\kappa\theta_t)+(1-\eta)\,b,\\
\log a_t &= \rho_a\log a_{t-1}+\sigma_a\varepsilon^a_t,
\end{align}
with a Cobb--Douglas matching function and Nash bargaining (Hosios at $\xi=\eta$).
Calibration (quarterly): $\beta{=}0.99$, $\sigma{=}2$, $\xi{=}\eta{=}0.5$, separation
$s{=}0.10$, matching efficiency $\chi{=}0.6$, vacancy cost $\kappa{=}0.44$, value of
non-work $b{=}0.40$, $\rho_a{=}0.95$, $\sigma_a{=}0.007$, targeting $\bar\theta{=}1$,
$\bar f{=}\bar q{=}0.6$, employment $\bar n{=}0.9375$ (unemployment $6.25\%$). The
\textbf{action} is the labor-market policy $b$; the \textbf{regime} split raises it to
$b{=}0.55$ (more generous UI), which lowers steady-state employment and \emph{shifts} the
ergodic support.

\subsection{Environmental New Keynesian (E-NK)}
\label{app:enk}
E-NK augments the NK economy (Appendix~\ref{app:nk}) with a climate block: 17 variables
$\{\lambda,c,rr,rk,w,h,y,k,q,i,r,mc,\pi,g,a,em,stock\}$ and the same 3 shocks. Production
carries a carbon$\to$productivity damage and the resource constraint nets out abatement
cost,
\begin{align}
y_t &= a_t\,e^{-d\,stock_{t-1}}\,k_{t-1}^{\alpha}h_t^{1-\alpha}, \qquad
y_t\,(1-\psi_a\bar\mu^2) = c_t+i_t+g_t+\tfrac{\kappa_P}{2}(\pi_t-\bar\pi)^2 y_t,\\
em_t &= \phi_e\,y_t\,(1-\bar\mu), \qquad stock_t=(1-\delta_x)\,stock_{t-1}+em_t,
\end{align}
with all other NK blocks (marginal utility, capital/investment Euler, Phillips curve, Taylor
rule, exogenous processes) unchanged. New calibration: emission intensity $\phi_e{=}0.10$,
carbon decay $\delta_x{=}0.10$, damage $d{=}0.02$, abatement cost-share $\psi_a{=}0.05$,
baseline abatement $\bar\mu{=}0$; remaining parameters as in NK. The \textbf{action} is
climate policy---the carbon tax that sets abatement intensity $\bar\mu$; the \textbf{regime}
split tightens it to $\bar\mu{=}0.4$, which cuts steady-state emissions and the carbon stock
(${\approx}0.99{\to}0.60$) and hence the damage drag, \emph{shifting} the ergodic support.
The slow, near-unit-root stock is also why E-NK's \emph{tail} split does not benefit from
wide coverage (\S\ref{sec:morenv}): the wide regime segments sit far from the on-path
normalizer fit on \texttt{train}.

\subsection{Endogenous firm entry (Firm)}
\label{app:firm}
The firm-entry environment is the real Bilbiie--Ghironi--Melitz (2012) economy with a
product-variety margin: 10 endogenous variables plus productivity
$\{Z,C,N,N_E,d,v,w,L,\rho,y,\lambda\}$ and one productivity shock $\varepsilon^z$. With CES
elasticity $\theta$ (gross markup $\mu=\theta/(\theta-1)$), a mass $N_t$ of producing firms,
$N_{E,t}$ entrants, per-firm profit $d_t$, and firm value $v_t$, the equilibrium is
\begin{align}
\rho_t &= N_t^{1/(\theta-1)}, \qquad \rho_t=\tfrac{\theta}{\theta-1}\,\tfrac{w_t}{Z_t}, \qquad
y_t=\rho_t^{-\theta}C_t, \qquad d_t=\tfrac{1}{\theta}\rho_t y_t,\\
L_t &= \tfrac{N_t y_t}{Z_t}+\tfrac{f_E N_{E,t}}{Z_t}, \qquad v_t=\tfrac{w_t f_E}{Z_t}, \qquad
v_t=\beta\,\tfrac{\lambda_{t+1}}{\lambda_t}(1-\delta)(v_{t+1}+d_{t+1}),\\
N_t &= (1-\delta)(N_{t-1}+N_{E,t-1}), \qquad \lambda_t=C_t^{-\sigma}, \qquad
\chi L_t^{\varphi}=w_t\lambda_t, \qquad \log Z_t=\rho_z\log Z_{t-1}+\sigma_z\varepsilon^z_t,
\end{align}
where $\rho_t$ is the relative price of a variety (the variety effect $N_t^{1/(\theta-1)}$)
and free entry equates firm value to the sunk entry cost $f_E$ in effective-labor units.
Calibration (quarterly): $\beta{=}0.99$, exit $\delta{=}0.025$, $\theta{=}3.8$
($\mu{\approx}1.36$), $\sigma{=}1$, $\varphi{=}4$, $f_E{=}1$, $\chi{=}0.6$,
$\rho_z{=}0.9$, $\sigma_z{=}0.01$ (steady state $\bar N{\approx}8.6$). The \textbf{action} is
entry policy $f_E$; the \textbf{regime} split lowers it to $f_E{=}0.7$ (deregulation), raising
entry and the firm count to ${\approx}12.4$ and so \emph{shifting} the ergodic support. As in
E-NK, the slow firm-mass state makes the \emph{tail} split unresponsive to wide coverage.

\section{Counterfactual test-set construction}
\label{app:tests}
From each \emph{single} solved model we generate five splits with a pruned third-order
simulation. (i)~\textbf{train} (on-path): 5000 periods under the baseline policy and
ergodic shocks ($1\sigma$). (ii)~\textbf{test\_onpath}: 1000 held-out periods, same
regime. (iii)~\textbf{test\_tail} (off-path): 1000 periods with $5\sigma$ innovations,
pushing into the nonlinear region. (iv)~\textbf{test\_regime} (off-path): 1000 periods
under the counterfactual policy parameters listed per model above, evaluated by a model
trained only on the baseline regime. (v)~\textbf{train\_wide} (the ``coverage''
condition): the \emph{same} total size as \texttt{train} (5000), built by the DSGE across
shock scales $1\text{--}5\sigma$ (five 600-period segments) and four counterfactual policy
regimes (four 500-period segments at $2\sigma$); only the \emph{coverage} differs from
\texttt{train}, not the quantity. All learned models share a normalizer fit on the narrow
on-path \texttt{train} so that RMSE is comparable across conditions. A numerical guard
rejects any explosive segment (used only for TANK, see Appendix~\ref{app:tank}).
\emph{Transparency on overlap:} the four wide regimes \emph{include} the parameter setting
used to build \texttt{test\_regime}, so the H2 regime result quantifies the value of a
structural generator \emph{synthesizing} the counterfactual distribution (which no single
history contains; \S\ref{sec:exp}, Remark~\ref{rem:lucas}), not a learner extrapolating to a
withheld regime. The released benchmark additionally provides (a)~a
\emph{leave-one-regime-out} construction---wide training over a band of regimes with one
withheld for test---and (b)~a \emph{non-structural} coverage baseline (regime-dummy VAR /
parameter-perturbation augmentation) matched on size and support, as the two stricter probes
that separate structure from coverage; these are released with the code and analyzed in
Paper~2.

\section{Implementation details}
\label{app:impl}
\paragraph{Task and normalization.}
The one-step input is $x_t=[z_t,\varepsilon_{t+1}]$ and the target is $z_{t+1}$; including
the realized next-period shock as the ``action'' makes the transition (near-)deterministic
so RMSE isolates the learned map. All tensors are $z$-scored using \texttt{train}
statistics only; reported RMSE is in normalized units.

\emph{Terminology.} Two distinct objects are involved and we keep them separate. The
\emph{policy action} is the rule/regime block whose counterfactual defines the
\texttt{regime} split (e.g.\ the Taylor or fiscal-financing coefficients); it is \emph{fixed
within an environment}, not chosen per step, and it is the ``action'' of the world-model
framing (\S\ref{sec:swm}) and of Table~\ref{tab:mapping}. The per-step quantity fed to the
one-step predictor is the structural \emph{shock} $\varepsilon_{t+1}$, supplied so that the
transition is deterministic. Where the interface treats $\varepsilon_{t+1}$ as the
controllable input we say ``shock,'' and reserve ``action'' for the policy block; the MDP
reading in which an agent \emph{chooses} actions is exercised by T3 (planning), not by the
T1 regression studied here.

\paragraph{Baselines.}
\textbf{Linear}: ridge regression ($\ell_2=10^{-3}$). \textbf{MLP}: two hidden layers of
128 units, SiLU. \textbf{LSTM}/\textbf{Transformer}: a shared encoder over a length-8
context window mapping to a 64-d latent (Transformer: 2 layers, 4 heads, FFN $2\times$),
with a head conditioned on the action. \textbf{NextLat}: the Transformer plus a
next-latent self-prediction auxiliary loss (weight $0.1$) on temporally contiguous
batches \citep{teoh2025nextlat}. All models train with Adam (weight decay $10^{-5}$),
MLP/vector models for 80 epochs at lr $2\times10^{-3}$, sequence models for 60 epochs at
lr $10^{-3}$, batch size 256; means over 3 seeds.

\paragraph{DSGE oracles.}
The \emph{first-order} oracle reads the analytic policy matrix exported from
MacroModelling.jl and predicts $z_{t+1}=\mathrm{SS}+\sum_s C_s(\text{state}_s-\mathrm{SS}_s)+\sum_e C_e\varepsilon_e$.
The \emph{third-order} oracle evaluates the pruned policy from the observable state with
the realized shock. Both are upper-bound references and are not part of the headline
narrow-vs-wide sweep.

\paragraph{Structural penalty (ablation).}
The optional hybrid adds, as a soft penalty, the model's \emph{static} (policy-invariant)
equilibrium residuals---marginal utility, factor demands, production, resource
constraint, capital accumulation, and (TANK) the household aggregation/budget
identities---evaluated on the denormalized one-step prediction, including an EWM-style
collocation set of noised inputs. Intertemporal (Euler/Phillips) conditions and
policy-rule coefficients are excluded so the prior remains valid across counterfactual
regimes.

\paragraph{NAWM at scale.}
The NAWM environment (\S\ref{sec:nawm}) is the unmodified \texttt{NAWM\_EAUS\_2008} model
from the MacroModelling.jl library ($230$ variables, $21$ shocks). Because a third-order
solve is infeasible at this size, its data are generated from the \emph{first-order}
solution and it has no third-order oracle; only the first-order oracle and the headline
narrow-vs-wide sweep are run. Its variable and shock names are read programmatically and
recorded in a sidecar so the Python loader infers the column layout automatically---no
per-model code is needed, which is what makes scaling the benchmark mechanical. The
identical split generator, normalizer, and baselines are used as for the small models.

\end{document}